\documentclass[11pt]{article}%

\usepackage{datetime}
\usepackage[margin=1in]{geometry}
\usepackage{amsmath}
\usepackage{amssymb}
\usepackage{mathptmx}
\usepackage{paralist}
\usepackage{graphicx}   
\usepackage{caption}
\usepackage{color}
\usepackage{xspace}

\usepackage[cmyk]{xcolor}

\usepackage{hyperref}%
\hypersetup{%
   breaklinks,%
   ocgcolorlinks,
   colorlinks=true,%
   linkcolor=[rgb]{0.45,0.0,0.0},%
   citecolor=[rgb]{0,0,0.45}
}

\usepackage{amsmath}%
\usepackage[amsmath,amsthm,thmmarks]{ntheorem}
\allowdisplaybreaks

\theoremseparator{.}%

\usepackage{titlesec}
\titlelabel{\thetitle. }

\numberwithin{figure}{section}
\numberwithin{table}{section}
\numberwithin{equation}{section}

\newtheorem{theorem}{Theorem}[section] 
\newtheorem{lemma}[theorem]{Lemma}

{\theorembodyfont{\rm}%
   \newtheorem{defn}[theorem]{Definition}%
   \newtheorem{observation}[theorem]{Observation}
   
}

\newlength{\savedparindent}

\providecommand{\pbrcx}[1]{\left[ {#1} \right]}

\renewcommand{\Re}{{\rm I\!\hspace{-0.025em} R}}

\newcommand{\eqlab}[1]{\label{equation:#1}}
\newcommand{\Eqref}[1]{Eq.~(\ref{equation:#1})}

\newcommand{\lemlab}[1]{\label{lemma:#1}}
\newcommand{\lemref}[1]{Lemma~\ref{lemma:#1}}

\newcommand{\obslab}[1]{\label{observation:#1}}
\newcommand{\obsref}[1]{Observation~\ref{observation:#1}}

\newcommand{\rlxlab}[1]{\label{relaxation:#1}}
\newcommand{\rlxref}[1]{Relaxation~\ref{relaxation:#1}}

\newcommand{\alglab}[1]{\label{algorithm:#1}}
\newcommand{\algref}[1]{Algorithm~\ref{algorithm:#1}}

\newcommand{\seclab}[1]{\label{section:#1}}
\newcommand{\secref}[1]{Section~\ref{section:#1}}
\newcommand{\thmlab}[1]{\label{theorem:#1}}
\newcommand{\thmref}[1]{Theorem~\ref{theorem:#1}}

\newcommand{\pth}[2][\!]{#1\left({#2}\right)}%

\newcommand{\brc}[1]{\left\{ {#1} \right\}}%
\newcommand{\sep}[1]{\,\left|\, {#1} \right.}%

\newcommand{\Ex}[1]{\mathop{\mathbf{E}}\pbrcx{#1}}
\newcommand{\Prob}[1]{\mathop{\mathbf{Pr}}\!\pbrcx{#1}}

\newcommand{\cardin}[1]{\left| {#1} \right|}%

\newcommand{\norm}[2]{\left|\left| {#2} \right|\right|_{#1}}

\newcommand{\dif}[1]{~\mathrm{d #1}}

\newcommand{\detm}[1]{\mathrm{det}\pth{#1}}

\newcommand{\E}{\mathbb{E}}      
\newcommand{\1}{\mathbf{1}}      

\newcommand{\cloud}{\mathsf{cloud}}

\newcommand{\iprod}[2]{\langle #1,#2 \rangle}
\newcommand{\normal}[2]{\mathcal{N}\pth{#1,#2}}

\newcommand{\eigmin}[1]{\lambda_{\mathrm{min}}\pth{#1}}

\DeclareMathOperator{\good}{\mathsf{good}}

\DeclareMathOperator{\bad}{\mathsf{bad}}

\begin{document}

\title{{\bf Approximate Hypergraph Coloring under \\ Low-discrepancy and Related Promises}}

\author{
Vijay V. S. P. Bhattiprolu\thanks{Supported by NSF CCF-1115525 \tt vpb@cs.cmu.edu} \and
Venkatesan Guruswami\thanks{Supported in part by NSF grant CCF-1115525. {\tt guruswami@cmu.edu} } \and
Euiwoong Lee\thanks{Supported by a Samsung Fellowship and NSF CCF-1115525. {\tt euiwoonl@cs.cmu.edu} }
}

\date{Computer Science Department \\ Carnegie Mellon University \\ Pittsburgh, PA 15213.}

\maketitle

\begin{abstract}

 A hypergraph is said to be $\chi$-colorable if its vertices can be colored with $\chi$ colors so that no hyperedge is monochromatic. $2$-colorability is a fundamental property (called Property B) of hypergraphs and is extensively studied in combinatorics. Algorithmically, however, given a $2$-colorable $k$-uniform hypergraph, it is NP-hard to find a $2$-coloring miscoloring fewer than a fraction $2^{-k+1}$ of hyperedges (which is trivially achieved by a random $2$-coloring), and the best algorithms to color the hypergraph properly require $\approx n^{1-1/k}$ colors, approaching the trivial bound of $n$ as $k$ increases.

 In this work, we study the complexity of approximate hypergraph coloring, for both the maximization (finding a $2$-coloring with fewest miscolored edges) and minimization (finding a proper coloring using fewest number of colors) versions, when the input hypergraph is promised to have the following stronger properties than $2$-colorability:
 \begin{itemize}
 \item Low-discrepancy: If the hypergraph has a $2$-coloring of discrepancy $\ell \ll \sqrt{k}$, we give an algorithm to color the hypergraph with $\approx n^{O(\ell^2/k)}$ colors.
  However, for the maximization version, we prove NP-hardness of finding a $2$-coloring miscoloring a smaller than $2^{-O(k)}$ (resp. $k^{-O(k)}$) fraction of the hyperedges when $\ell = O(\log k)$ (resp. $\ell=2$). Assuming the Unique Games conjecture, we improve the latter hardness factor to $2^{-O(k)}$ for almost discrepancy-$1$ hypergraphs.

 \item Rainbow colorability: If the hypergraph has a $(k-\ell)$-coloring such that each hyperedge is polychromatic with all these colors (this is stronger than a $(\ell+1)$-discrepancy $2$-coloring), we give a $2$-coloring algorithm that miscolors at most $k^{-\Omega(k)}$ of the hyperedges when $\ell \ll \sqrt{k}$, and complement this with a matching Unique Games hardness result showing that when $\ell =\sqrt{k}$, it is hard to even beat the $2^{-k+1}$ bound achieved by a random coloring.
 \end{itemize}
\end{abstract}

\maketitle
\thispagestyle{empty}

\newpage

\section{Introduction}
Coloring (hyper)graphs is one of the most important and well-studied
tasks in discrete mathematics and theoretical computer science.  A
$k$-uniform hypergraph $G = (V, E)$ is said to be $\chi$-colorable if
there exists a coloring $c : V \mapsto \{ 1, \dots, \chi \}$ such that
no hyperedge is monochromatic, and such a coloring $c$ is referred to as 
a proper $\chi$-coloring.
 Graph and hypergraph coloring has been
the focus of active research in both fields, and has served as the
benchmark for new research paradigms such as the probabilistic method 
(Lov\'{a}sz local lemma~\cite{EL75}) and semidefinite programming 
(Lov\'{a}sz theta function~\cite{Lovasz79}).

While such structural results are targeted towards special classes of
hypergraphs, given a general $\chi$-colorable $k$-uniform hypergraph,
the problem of reconstructing a $\chi$-coloring is known to be a hard
task.  Even assuming $2$-colorability, reconstructing a proper $2$-coloring
is a classic NP-hard problem for $k \ge 3$. Given the intractability
of proper $2$-coloring, two notions of {\em approximate coloring} of
$2$-colorable hypergraphs have been studied in the literature of
approximation algorithms.  The first notion, called {\em
Min-Coloring}, is to minimize the number of colors while
still requiring that every hyperedge be non-monochromatic.  The second
notion, called {\em Max-2-Coloring} allows only $2$ colors, but the
objective is to maximize the number of non-monochromatic
hyperedges.\footnote{The maximization version is also known as
Max-Set-Splitting, or more specifically Max $k$-Set-Splitting when
considering $k$-uniform hypegraphs, in the literature.}

Even with these relaxed objectives, the promise that the input
hypergraph is $2$-colorable seems grossly inadequate for polynomial
time algorithms to exploit in a significant way.  For Min-Coloring, given 
a $2$-colorable $k$-uniform hypergraph, the best known algorithm uses 
$O(n^{1- \frac{1}{k}})$ colors~\cite{CF96, AKMH96}, which tends to the
trivial upper bound $n$ as $k$ increases.  This problem has been 
actively studied from the hardness side, motivating many new developments 
in constructions of probabilistically checkable proofs. Coloring 
$2$-colorable hypergraphs with $O(1)$ colors was shown to be NP-hard for 
$k \ge 4$ in \cite{GHS02} and $k=3$ in \cite{DRS02}. An exciting body of 
recent work has pushed the hardness beyond poly-logarithmic 
colors~\cite{DG13,GHHSV14,KS14b,Huang15}. In particular, \cite{KS14b} 
shows quasi-NP-hardness of $2^{(\log n)^{\Omega(1)}}$-coloring a 
$2$-colorable hypergraphs (very recently the exponent  was shown to 
approach $1/4$ in \cite{Huang15}).

The hardness results for Max-2-Coloring show an even more pessimistic 
picture, wherein the naive {\em random assignment} (randomly give
one of two colors to each vertex independently to leave a
$(\frac{1}{2})^{k - 1}$ fraction of hyperedges monochromatic in
expectation), is shown to have the best guarantee for a polynomial time
algorithm when $k \geq 4$ (see~\cite{Hastad01}).

Given these strong intractability results, it is natural to consider
what further relaxations of the objectives could lead to efficient
algorithms.  For maximization versions, Austrin and H\aa
stad~\cite{AH13} prove that (almost\footnote{We say a hypergraph is
{\em almost} $\chi$-colorable for a small constant $\epsilon > 0$,
there is a $\chi$-coloring that leaves at most $\epsilon$ fraction of
hyperedges monochromatic.}) $2$-colorability is useless (in a formal
sense that they define) for any Constraint Satisfaction Problem (CSP)
that is a relaxation of $2$-coloring~\cite{Wenner14}.  Therefore, it
seems more natural to find a stronger promise on the hypergraph than
mere $2$-colorability that can be significantly exploited by
polynomial time coloring algorithms for the objectives of Min-Coloring
and Max $2$-Coloring. This motivates our main question {\em ``how
strong a promise on the input hypergraph is required for polynomial
time algorithms to perform significantly better than naive algorithms
for Min-Coloring and Max-2-Coloring?'' }

There is a very strong promise on $k$-uniform hypergraphs which makes
the task of proper $2$-coloring easy. If a hypergraph is $k$-partite (i.e.,
there is a $k$-coloring such that each hyperedge has each color
exactly once), then one can properly $2$-color the hypergraph in polynomial
time. The same algorithm can be generalized to hypergraphs which admit a
$c$-balanced coloring (i.e., $c$ divides $k$ and there is a
$k$-coloring such that each hyperedge has each color exactly
$\frac{k}{c}$ times). This can be seen by random hyperplane rounding of 
a simple SDP, or even simpler by solving a homogeneous linear system and 
iterating~\cite{alon-pc}, or by a random recoloring method analyzed using 
random walks~\cite{mcdiarmid}. In fact, a proper $2$-coloring can be 
efficiently achieved assuming that the hypergraph admits a fair partial 
$2$-coloring, namely a pair of disjoint subsets $A$ and $B$ of the 
vertices such that for every hyperedge $e$, 
$|e \cap A| = |e \cap B| > 0$~\cite{mcdiarmid}.

The promises on structured colorings that we consider in this work are
natural relaxations of the above strong promise of a
perfectly balanced (partial) coloring.

\begin{itemize}
\item A hypergraph is said to have {\em discrepancy $\ell$} when there is a 
$2$-coloring such that in each hyperedge, the difference between the number 
of vertices of each color is at most $\ell$.

\item A $\chi$-coloring $(\chi \leq k)$ is called {\em rainbow} if every 
hyperedge contains each color at least once. 

\item A $\chi$-coloring $(\chi \geq k)$ is called {\em strong} if every 
hyperedge contains $k$ different colors. 
\end{itemize}

These three notions are interesting in their own right, and have been
independently studied. 
Discrepancy minimization has recently seen different algorithmic 
ideas~\cite{Bansal10, LM12, Rothvoss14} to give constructive proofs of the 
classic six standard deviations result of Spencer~\cite{Spencer85}.
Rainbow coloring admits a natural interpretation as a partition of $V$ into 
the maximum number of disjoint vertex covers, and has been actively studied 
for geometric hypergraphs due to its applications in sensor networks~\cite{BPRS13}.
Strong coloring is closely related to graph coloring by definition, and is 
known to capture various other notions of coloring~\cite{AH05}.
It is easy to see that $\ell$-discrepancy ($\ell < k$), $\chi$-rainbow
colorability ($2 \leq \chi \leq k)$, and $\chi$-strong colorability
$(k \leq \chi \leq 2k - 2)$ all imply $2$-colorability.  For odd $k$, both 
$(k + 1)$-strong colorability and $(k-1)$-rainbow colorability imply 
discrepancy-$1$, so strong colorability and rainbow colorability seem stronger 
than low discrepancy.

Even though they seem very strong, previous works have mainly focused
on hardness with these promises.  The work of Austrin et
al.~\cite{AGH14} shows NP-hardness of finding a proper $2$-coloring under the
discrepancy-$1$ promise.  The work of Bansal and Khot~\cite{BK10}
shows hardness of $O(1)$-coloring even when the input hypergraph is
promised to be almost $k$-partite (under the Unique Games Conjecture);
Sachdeva and Saket~\cite{SS13} establish NP-hardness of
$O(1)$-coloring when the graph is almost $k/2$-rainbow colorable; and
Guruswami and Lee~\cite{GL15a} establish NP-hardness when the graph is
perfectly (not almost) $\frac{k}{2}$-rainbow colorable, or admits a
$2$-coloring with discrepancy $2$.
These hardness results indicate that it is still a nontrivial task to 
exploit these strong promises and outperform naive algorithms.

\subsection{Our Results}
In this work, we prove that our three promises, unlike mere $2$-colorability, 
give enough structure for polynomial time algorithms to perform significantly 
better than naive algorithms. We also study these promises from a hardness 
perspective to understand the asymptotic threshold at which beating naive 
algorithms goes from easy to UG/NP-Hard. In particular assuming the UGC, for 
Max-$2$-Coloring under $\ell$-discrepancy or $k - \ell$-rainbow colorability, this 
threshold is $\ell = \Theta(\sqrt{k})$.

\begin{theorem}
There is a randomized polynomial time algorithm that produces a $2$-coloring of 
a $k$-uniform hypergraph $H$ with the following guarantee.
For any $0 < \epsilon < \frac{1}{2}$ (let $\ell = k^{\epsilon}$), there exists a 
constant $\eta > 0$ such that if $H$ is $(k - \ell)$-rainbow colorable or $(k + \ell)$-strong 
colorable, the fraction of monochromatic edges in the produced $2$-coloring is 
$O((\frac{1}{k})^{\eta k})$ in expectation. 
\end{theorem}
Our results indeed show that this algorithm significantly outperforms the random 
assignment even when $\ell$ approaches $\sqrt{k}$ asymptotically. 
See~\thmref{analysis:strong} and~\thmref{analysis:rainbow} for the precise statements.

For the $\ell$-discrepancy case, we observe that when $\ell < \sqrt{k}$, the framework of the 
second and the third authors~\cite{GL15b} yields an approximation algorithm that marginally 
(by an additive factor much less than $2^{-k}$) outperforms the random assignment, but 
we do not formally prove this here.

The following hardness results suggest that this gap between low-discrepancy and 
rainbow/strong colorability might be intrinsic. Let the term {\em UG-hardness} denote 
NP-hardness assuming the Unique Games Conjecture.

\begin{theorem}
\thmlab{hardness1}
For sufficiently large odd $k$, given a $k$-uniform hypergraph which admits a $2$-coloring 
with at most a $(\frac{1}{2})^{6k}$ fraction of edges of discrepancy larger than $1$, it 
is UG-hard to find a $2$-coloring with a $(\frac{1}{2})^{5k}$ fraction of monochromatic edges.
\end{theorem}

\begin{theorem}
\thmlab{hardness2}
For even $k \geq 4$, given a $k$-uniform hypergraph which admits a $2$-coloring with no edge 
of discrepancy larger than $2$, it is NP-hard to find a $2$-coloring with a $k^{-O(k)}$ 
fraction of monochromatic edges.
\end{theorem}

\begin{theorem}
\thmlab{hardness3}
For $k$ sufficiently large, given a $k$-uniform hypergraph which admits a $2$-coloring with no 
edge of discrepancy larger than $O(\log k)$, it is NP-hard to find a $2$-coloring with a 
$2^{-O(k)}$ fraction of monochromatic edges.
\end{theorem}

\begin{theorem}
\thmlab{hardness4}
For $k$ such that $\chi := k - \sqrt{k}$ is an integer greater than $1$, and any $\epsilon > 0$, 
given a $k$-uniform hypergraph which admits a $\chi$-coloring with at most $\epsilon$ fraction 
of non-rainbow edges, it is UG-hard to find a $2$-coloring with a $(\frac{1}{2})^{k - 1}$ 
fraction of monochromatic edges.
\end{theorem}

For Min-Coloring, all three promises lead to an $\tilde{O}(n^{\frac{1}{k}})$-coloring that 
is decreasing in $k$. These results are also notable in the sense that our promises are helpful 
not only for structured SDP solutions, but also for combinatorial {\em degree reduction} 
algorithms.

\begin{theorem}
    Consider any $k$-uniform hypergraph $H=(V,E)$ with $n$ vertices and $m$ edges.
    For any $\ell < O(\sqrt{k})$, 
    If $H$ has discrepancy-$\ell$, $(k - \ell)$-rainbow colorable, or $(k + \ell)$-strong colorable, 
    one can color $H$ with $\tilde{O}((\frac{m}{n})^{\frac{\ell^2}{k^2}}) \leq 
    \tilde{O}(n^{\frac{\ell^2}{k}})$ colors. 
\end{theorem} 

These results significantly improve the current best $\tilde{O}(n^{1 - \frac{1}{k}})$ colors 
that assumes only $2$-colorability. 
Our techniques give slightly better results depending on the promise --- see~\thmref{min:coloring}. 
Table~\ref{tab:summary} summarizes our results.

\begin{table}[htpb]
\centering
\small
{
\begin{tabular}{| c  | | c | c | c |}
\hline
Promises & $\ell$-Discrepancy & $(k - \ell)$-Rainbow & $(k + \ell)$-Strong \\
\hline
\hline
Max-2-Coloring &
$1 - (1/2)^{k - 1} + \delta$, \quad $\ell < \sqrt{k}^{\dagger}$ & 
$1 - (1/k)^{\Omega(k)}$, \quad $\ell \ll \sqrt{k}$ & 
$1 - (1/k)^{\Omega(k)}$, \quad $\ell \ll  \sqrt{k}$ 
\\
Algorithm &&&
\\
&&&
\\
\hline
Max-2-Coloring 
& UG: $1 - (1/2)^{5k}$, \,\,$\ell = 1$.
& UG: $1 - (1/2)^{k - 1}$, \,$\ell = \Omega(\sqrt{k})$
& \\
Hardness
& NP: $1 - (1/k)^{O(k)}$, \,\,$\ell = 2$.
& & \\
& NP: $1 - (1/2)^{O(k)}$, \,\,$\ell = \Omega(\log k)$
& & \\
& UG: $1 - (1/2)^{k - 1}$, \,\,$\ell \geq \sqrt{k}^{\dagger}$
& & \\
&&&
\\
\hline
Min-Coloring & $n^{\ell^2/k}$, \quad $\ell = O(\sqrt{k})$ 
& $n^{\ell^2/k}$, \quad $\ell = O(\sqrt{k})$ 
& $n^{\ell^2/k}$, \quad $\ell = O(\sqrt{k})$     
\\
Algorithm &&&
\\
\hline
\end{tabular}
}

\caption{{\small Summary of our algorithmic and hardness results with 
valid ranges of $\ell$. Two results with $\dagger$ are implied in~\cite{GL15b}.
The numbers of the first row indicate lower bounds on the fraction of 
non-monochromatic edges in a $2$-coloring produced by our algorithms.
$\delta := \delta(k, \ell) > 0$ is a small constant. The second row shows upper 
bounds on the fraction of non-monochromatic edges achieved by polynomial time 
algorithms. For the UG-hardness results, note that the input hypergraph does 
not have all edges satisfying the promises but almost edges satisfying them. 
The third row shows the upper bound upto $\log$ factors, on the number of 
colors one can use to properly $2$-color the graph.
}
\label{tab:summary}
}
\end{table}

\subsection{Techniques}
Our algorithms for Max-$2$-Coloring are straightforward applications of 
semidefinite programming, 
namely, we use natural vector relaxations of the promised properties, and 
round using a random hyperplane. The analysis however, is highly non-trivial 
and boils down to approximating a multivariate Gaussian integral. 
In particular, we show a (to our knowledge, new) upper bound on the Gaussian 
measure of simplicial cones in terms of simple properties of these cones. We 
should note that this upper bound is sensible only for simplicial cones that are 
well behaved with respect to the these properties. 
(The cones we are interested in are those given by the intersection of
hyperplanes whose normal vectors constitute a solution to our vector
 relaxations). We believe our analysis to be of independent interest as similar approaches may work for other $k$-CSP's.

\subsubsection{Gaussian Measure of Simplicial Cones}
As can be seen via an observation of Kneser \cite{kneser1936simplexinhalt}, the Gaussian measure of a simplicial cone is equal to the fraction of spherical volume taken up by a spherical simplex (a spherical simplex is the intersection of a simplicial 
cone with a ball centered at the apex of the cone). This however, is a very old problem in spherical geometry, and while some things are known, like a nice differential formula due to Schlafli (see \cite{schlafli1858multiple}), closed forms upto four dimensions (see \cite{murakami2005volume}), and a complicated power series expansion due to Aomoto \cite{aomoto1977analytic}, it is likely hopeless to achieve a closed form solution 
or even an asymptotic formula for the volume of general spherical simplices. 

Zwick \cite{zwick1998approximation} considered the performance of hyperplane rounding in various $3$-CSP formulations, and this involved analyzing the volume of a $4$-dimensional spherical simplex.  Due to the complexity of this volume function, the analysis was tedious, and non-analytic for many of the formulations. 
His techniques were based on the Schlafli differential formula, which relates the volume differential of a spherical simplex to the volume functions of its codimension-$2$ faces and dihedral angles. However, to our knowledge not much is known about the general volume function in even $6$ dimensions. This suggests that Zwick's techniques are unlikely to be scalable to higher 
dimensions.

On the positive side, an asymptotic expression is known in the case of symmetric spherical simplices, due to H. E. Daniels \cite{rogers1964packing} who gave the analysis for regular cones of angle $\cos^{-1}(1/2)$. His techniques were extended by Rogers \cite{rogers1961asymptotic} and Boeroeczky and Henk \cite{boroczky1999random} to the whole class of regular cones. 

We combine the complex analysis techniques employed by Daniels with a 
lower bound on quadratic forms in the positive orthant, to give an 
upper bound on the Gaussian measure of a much larger class of simplicial 
cones. 

\subsubsection{Column Subset Selection}
Informally, the cones for which our upper bound is relevant 
are those that are high dimensional in a strong sense, i.e. the normal 
vectors whose corresponding hyperplanes form the cone, must be such that 
no vector is too close to the linear span of any subset of the 
remaining vectors. 

When the normal vectors are solutions to our rainbow colorability SDP relaxation, this need not be true. However, this can be remedied. 
We consider the column matrix of these normal vectors, and using 
spectral techniques, we show that there is a reasonably large subset of columns
(vectors) that are well behaved with respect to condition number. 
We are then able to apply our Gaussian Measure bound to the cone given by 
this subset, admittedly in a slightly lower dimensional space.

\section{Approximate Max-$2$-Coloring}
\seclab{maxcolor:approx}

In this section we show how the properties of $(k+\ell)$-strong colorability and 
$(k-\ell)$-rainbow colorability in $k$-uniform hypergraphs allow one to $2$-color 
the hypergraph, such that the respective fractions of monochromatic edges are 
small. For 
$\ell = o(\sqrt{k})$, these guarantees handsomely beat the 
naive random algorithm (color every vertex blue or red uniformly and 
independently at random), wherein the expected fraction of monochromatic edges 
is $1/2^{k-1}$.

Our algorithms are straightforward applications of semidefinite programming, 
namely, we use natural vector relaxations of the above properties, and 
round using a random hyperplane. The analysis however, is quite involved.

\subsection{Semidefinite Relaxations}
\seclab{relaxations}
Our SDP relaxations for low-discrepancy, rainbow-colorability, and strong-colorability are the following.
Given that $\langle v_i, v_j \rangle = \frac{-1}{\chi - 1}$ when unit vectors $v_1, \dots, v_{\chi}$ form 
a $\chi$-regular simplex centered at the origin, it is easy to show that they are valid relaxations. 

\paragraph{Discrepancy $\ell$.}
\begin{align}
\rlxlab{discrepancy}
\forall\,e\in E,~\norm{2}{\sum_{i\in e} u_{i}} &\leq \ell \\
\forall\,i \in [n],~\norm{2}{u_i} &= 1 \nonumber\\
\forall\,i \in [n],~u_i\in\Re^n \nonumber
\end{align}

\paragraph{Feasibility.}
For $k,\ell$ such that $(k-\ell)\mod 2 \equiv 0$, consider any $k$-uniform hypergraph 
$H=(V=[n],E)$, and any $2$-coloring of $H$ of discrepancy $\ell$. 
Pick any unit vector $w\in\Re^n$. For each vertex of the first color in the 
coloring, assign the vector $w$, and for each vertex of the second color assign 
the vector $-w$. This is a feasible assignment, and hence \rlxref{discrepancy} 
is a feasible relaxation for any hypergraph of discrepancy $\ell$.

\paragraph{$(k-\ell)$-Rainbow Colorability.}

\begin{align}
\rlxlab{rainbow}
\forall\,e\in E,~\norm{2}{\sum_{i \in e} u_{i}} &\leq \ell \\
\forall\,e\in E,~\forall\,i<j \in e,~\iprod{u_{i}}{u_{j}} 
&\geq \frac{-1}{k-\ell-1} \nonumber\\
\forall\,i \in [n],~\norm{2}{u_i} &= 1 \nonumber\\
\forall\,i \in [n],~u_i\in\Re^n \nonumber
\end{align}

\paragraph{Feasibility.}
Consider any $k$-uniform hypergraph $H=(V=[n],E\subseteq {V\choose k})$, and 
any $(k-\ell)$-rainbow coloring of $H$. As testified by the vertices of the 
$(k-\ell)$-simplex, we can always choose unit vectors $w_1 \dots w_{k-\ell}\in\Re^n$ 
satisfying,
\[
\forall\,i<j \in [k-\ell],~\iprod{w_{i}}{w_{j}} = \frac{-1}{k-\ell-1},
\]
It is not hard to verify that consequently,
\[
\forall\,a_1,\dots ,a_{k-\ell} \in [l],~\sum_{i\in [k-\ell]} a_i =k,\quad
\text{we have,~}
~\norm{2}{\sum_{i \in e} a_iw_{i}} \leq \ell
\]
\medskip

\noindent
For each vertex of the color $i$, assign the vector $w_i$. This is a feasible 
assignment, and hence \rlxref{rainbow} is a feasible relaxation for any hypergraph of 
rainbow colorability $k-\ell$.
\bigskip

\paragraph{$(k+\ell)$-Strong Colorability.}

\begin{align}
\rlxlab{strong}
\forall\,e\in E,~\forall\,i<j \in e,\quad\iprod{u_{i}}{u_{j}} &= \frac{-1}{k+\ell-1} \\
\forall\,i \in [n],~\norm{2}{u_i} &= 1 \nonumber\\
\forall\,i \in [n],~u_i\in\Re^n \nonumber
\end{align}

\paragraph{Feasibility.}
Consider any $k$-uniform hypergraph $H=(V=[n],E\subseteq {V\choose k})$, and 
any $(k+\ell)$-strong coloring of $H$. As testified by the vertices of the 
$(k+\ell)$-simplex, we can always choose unit vectors $w_1 \dots w_{k+\ell}\in\Re^n$ 
satisfying,
\[
    \forall\,i<j \in [k-\ell],~\iprod{w_{i}}{w_{j}} = -\frac{1}{k+\ell-1},
\]
It is not hard to verify that consequently,
\[
    \forall\,J\subset [k+\ell],\,\cardin{J}=k,~\norm{2}{\sum_{i\in J} w_{i}} = \ell
\]
\medskip

\noindent
For each vertex of the color $i$, assign the vector $w_i$. This is a feasible 
assignment, and hence the \rlxref{strong} is a feasible relaxation for any 
hypergraph of strong colorability $k+\ell$.
\bigskip

Our rounding scheme is the same for all the above relaxations. 
\paragraph{Rounding Scheme.}
Pick a standard $n$-dimensional Gaussian random vector $r$. For any $i\in [n]$, 
if $\iprod{v_i}{r}\geq 0$, then vertex $i$ is colored blue, and otherwise it is 
colored red.

\subsection{Setup of Analysis}
\seclab{setup:local}
We now setup the framework for analyzing all the above relaxations.

\noindent
Consider a standard $n$-dimensional Gaussian random vector $r$, i.e. each coordinate is 
independently picked from the standard normal distribution $\normal{0}{1}$. The following 
are well known facts (the latter being due to Renyi), 

\begin{lemma}
\lemlab{sphere:gauss}
    $r/\norm{2}{r}$ is uniformly distributed over the unit sphere in $\Re^n$.
\end{lemma}

\paragraph{Note.}
\lemref{sphere:gauss} establishes that our rounding scheme is equivalent to random 
hyperplane rounding.

\begin{lemma}
\lemlab{gauss:affine:equiv}
    Consider any $j< n$. The projections of $r$ onto the pairwise orthogonal unit 
    vectors $e_1,\dots ,e_j$ are independent and have distribution $\normal{0}{1}$.
\end{lemma}

Next, consider any $k$-uniform hypergraph $H=(V=[n],E\subseteq {V\choose k})$ that is 
feasible for any of the aforementioned formulations. Our goal now, is to analyze the 
expected number of monochromatic edges. To obtain this expected fraction with high 
probability, we need only repeat the rounding scheme polynomially many times, and the 
high probability of a successful round follows by Markov's inequality. Thus we are 
only left with bounding the probability that a particular edge is monochromatic. 

To this end, consider any edge $e\in E$ and let the vectors 
corresponding to the vertices in $e$ be $u'_1,\dots,u'_k$. Consider a  
$k$-flat $\mathcal{F}$ (subspace of $\Re^n$ congruent to $\Re^k$), containing 
$u'_1,\dots,u'_k$. Applying \lemref{gauss:affine:equiv} to the standard basis of 
$\mathcal{F}$, implies that the projection of $r$ into $\mathcal{F}$ has the standard 
$k$-dimensional Gaussian distribution. Now since projecting $r$ onto 
$\mathrm{Span}\pth{u'_1,\dots u'_k}$ preserves the inner products 
$\brc{\iprod{r}{u'_i}}_i$~, we may assume without loss of generality that 
$u'_1,\dots ,u'_k$ are vectors in $\Re^k$, and the rounding scheme corresponds to 
picking a random $k$-dimensional Gaussian vector $r$, and proceeding as before. 
\medskip

\noindent
Let $U$ be the $k\times k$ matrix whose columns are the vectors $u'_1,\dots,u'_k$ 
and $\mu$ represent the Gaussian measure in $\Re^k$.
Then the probability of $e$ being monochromatic in the rounding is given by, 
\begin{equation}
\eqlab{monochr:prob}
    \mu\pth{\brc{x\in \Re^k \sep{U^Tx \geq 0}}}
    +
    \mu\pth{\brc{x\in \Re^k \sep{U^Tx < 0}}}
    =
    2 \mu\pth{\brc{x\in \Re^k \sep{U^Tx \geq 0}}}
\end{equation}
In other words, this boils down to analyzing the Gaussian measure of the cone given 
by $U^Tx\geq 0$. We thus take a necessary detour.

\subsection{Gaussian Measure of Simplicial Cones}
In this section we show how to bound the Gaussian measure of a special class of simplicial cones. 
This is one of the primary tools in our analysis of the previously introduced SDP relaxations. 
We first state some preliminaries.

\subsubsection{Preliminaries}

\paragraph{Simplicial Cones and Equivalent Representations.}
A simplicial cone in $\Re^k$, is given by the intersection of a set of $k$ linearly independent 
halfspaces. For any simplicial cone with apex at position vector $p$, there is a unique set 
(upto changes in lengths), of $k$ linearly independent vectors, such that the direct sum of $\brc{p}$ 
with their positive span produces the cone. 
Conversely, a simplicial cone given by the direct sum of $\brc{p}$ and the positive span of $k$ linearly 
independent vectors, can be expressed as the intersection of a unique set of $k$ halfspaces with apex 
at $p$. We shall refer to the normal vectors of the halfspaces above, as simply normal vectors of the cone, 
and we shall refer to the spanning vectors above, as simplicial vectors.
We represent a simplicial cone $C$ with apex at $p$, as $(p,U,V)$ where $U$ is a column matrix 
of unit vectors $u_1, \dots ,u_k$ (normal vectors), $V$ is a column matrix of unit vectors 
$v_1, \dots ,v_k$ (simplicial vectors) and 
\[
C = \brc{x\in\Re^k \sep{u_1^T x\geq p_1,\dots ,u_k^T x \geq p_k}} = 
\brc{p+x_1v_1+\dots +x_kv_k \sep{x\geq 0, x\in\Re^k}}
\]

\paragraph{Switching Between Representations.}
Let $C\equiv (0,U,V)$ be a simplicial cone with apex at the origin. It is not hard to see that any $v_i$ 
is in the intersection of exactly $k-1$ of the $k$ halfspaces determined by $U$, and it is thus orthogonal 
to exactly $k-1$ vectors of the form $u_j$. We may assume without loss of generality that for any $v_i$, 
the only column vector of $U$ not orthogonal to it, is $u_i$. Thus clearly $V^TU = D$ where $D$ is some 
non-singular diagonal matrix. Let $A_U=U^TU$ and $A_V=V^TV$, be the gram matrices of the vectors. 
$A_U$ and $A_V$ are positive definite symmetric matrices with diagonal
entries equal to one (they comprise of the pairwise inner products of the normal and simplicial vectors respectively). 
Also, clearly, 
\begin{equation}
\eqlab{conversion}
V = U^{-T}D,
\qquad
A_V = DA_U^{-1}D
\end{equation}
One then immediately obtains:
$(A_V)_{ij} = \frac{a_{ij}}{\sqrt{a_{ii}a_{jj}}}$, and 
$
(A_U)_{ij} = \frac{-a'_{ij}}{\sqrt{a'_{ii}a'_{jj}}}$.
where $a_{ij}$ and $a'_{ij}$ are the cofactors of the $(i,j)^{th}$ entries of  $A_U$ and $A_V$ respectively. 

\paragraph{Formulating the Integral.}
Let $C\equiv (0,U,V)$ be a simplicial cone with apex at the origin, and for $x\in\Re^k$, let$\dif{x}$ denote the 
differential of the standard $k$-dimensional Lebesgue measure. Then the Gaussian measure of $C$ is given by,
\begin{alignat*}{2}
&~\frac{1}{\pi^{k/2}}\int\limits_{U^Tx\geq 0} e^{-\norm{2}{x}^2}\dif{x}\nonumber
&&=~
\frac{\mathrm{det}(V)}{\pi^{k/2}}\int\limits_{\Re^k_{+}} e^{-\norm{2}{Vx}^2}\dif{x}
\qquad \text{Subst. $x\leftarrow Vx$}\nonumber\\
=~&
\frac{\mathrm{det}(V)}{\pi^{k/2}}\int\limits_{\Re^k_{+}} e^{-\norm{2}{U^{-T}Dx}^2}\dif{x}
\qquad \text{By \Eqref{conversion}} \qquad \nonumber
&&=~
\frac{\mathrm{det}(V)}{\pi^{k/2}~\mathrm{det}(D)}
\int\limits_{\Re^k_{+}} e^{-\norm{2}{U^{-T}x}^2}\dif{x}
\qquad \text{Subst. $x\leftarrow Dx$}\nonumber\\
=~&
\frac{1}{\pi^{k/2}~\mathrm{det}(U)}\int\limits_{\Re^k_{+}} e^{-\norm{2}{U^{-T}x}^2}\dif{x}\nonumber
\eqlab{formulation}
&&=~
\frac{1}{\pi^{k/2}\sqrt{\detm{A_U}}}\int\limits_{\Re^k_{+}} e^{-x^{T}A_U^{-1}x}\dif{x}
\end{alignat*}

For future ease of use, we give a name to some properties.

\begin{defn}
    The \emph{para-volume} of a set of vectors (resp. a matrix $U$), is the volume of 
    the parallelotope determined by the set of vectors (resp. the column vectors of $U$).
\end{defn}

\begin{defn}
    The \emph{sum-norm} of a set of vectors (resp. a matrix $U$), is the length
    of the sum of the vectors (resp. the sum of the column vectors of $U$).
\end{defn}
 
\paragraph{Walkthrough of Symmetric Case Analysis.}
We next state some simple identities that can be found in say, \cite{rogers1964packing}, some of which were originally used 
by Daniels to show that the Gaussian measure of a symmetric cone in $\Re^k$ of angle 
$\cos^{-1}(1/2)$ (between any two simplicial vectors) is 
$
\frac{(1+o(1))~e^{k/2-1}}{\sqrt{2}^{k+1}\sqrt{k}^{k-1}\sqrt{\pi}^{k}}.
$
We state these identities, while loosely describing the analysis of the symmetric case, to give the reader an idea of their purpose.

\noindent
First note that the gram matrices $S_U$ and $S_V$, of the symmetric cone of angle 
$\cos^{-1}(1/2)$ are given by:
\[
S_U = (1+1/k)\mathrm{I} -\1\1^T/k
\qquad
S_V = (\mathrm{I} + \1\1^T)/2
\]
Thus $x^T S_U^{-1} x$ is of the form, 
\begin{equation}
\eqlab{sym:quad}
\alpha\norm{1}{x}^2 + \beta\norm{2}{x}^2
\end{equation}

The key step is in linearizing the $\norm{1}{x}^2$ term in the exponent, which allows us to 
separate the terms in the multivariate integral into a product of univariate integrals, and this is 
easier to analyze.
\begin{lemma}[Linearization]
\lemlab{linearize}
$
    \sqrt{\pi}e^{-s^2} = \int_{-\infty}^{\infty} e^{-t^2 + 2its}\dif{t}
$
\end{lemma}

\begin{observation}
\obslab{abs:val}
    Let $f:\pth{-\infty,\infty}\mapsto \mathbb{C}$ be a continuous complex function. Then, 
$
    \cardin{\,\,\int\limits_{-\infty}^{\infty} f(t) \dif{t}\,} 
    \leq 
    \int\limits_{-\infty}^{\infty} \cardin{f(t)} \dif{t}.
$\end{observation}

On applying \lemref{linearize} to \Eqref{formulation} in the symmetric case, one obtains a product 
of identical univariate complex integrals. Specifically, by \Eqref{formulation}, \Eqref{sym:quad}, 
and \lemref{linearize}, we have the expression, 
\begin{align*} 
\int\limits_{\Re_+^k} e^{ - \beta\norm{2}{x}^2 -\alpha\norm{1}{x}^2} \dif{x}
~=
\int\limits_{-\infty}^{\infty} e^{-t^2}\int\limits_{\Re_+^k} e^{-\beta (x_1^2 +\dots \,x_k^2 ) \,+\, 
2it\sqrt{\alpha}(x_1 +\dots \,x_k)} \dif{x}\dif{t} 
=~
\int\limits_{-\infty}^{\infty} e^{-t^2}\pth{\int\limits_{0}^{\infty} 
e^{-\beta s^2 + 2it\sqrt{\alpha}s} \dif{s}}^{k}
\end{align*}
The inner univariate complex integral is not readily evaluable. To circumvent this, 
one can change the line  of integration so as to shift mass form the inner integral to the outer integral. 
Then we can apply the crude upper bound of \obsref{abs:val} to the inner integral, and by design, the error in our estimate is small.

\begin{lemma}[Changing line of integration]
\lemlab{change:line}
    Let $g(t)$ be a real valued function for real $t$. If, when interpreted as a complex function 
    in the variable $t=a+ib$, $g(a+ib)$ is an entire function, and furthermore, 
    $\lim\limits_{a\rightarrow \infty}g(a+ib) = 0$ for some fixed $b$, then we have, 
$        \int_{-\infty}^{\infty} g(t)\dif{t} = \int_{-\infty}^{\infty} g(a+ib)\dif{a}.
$
\end{lemma}

\paragraph{Squared $\mathrm{L}_1$ Inequality.}  
Motivated by the above linearization technique, we prove the following lower bound 
on quadratic forms in the positive orthant:

\begin{lemma}
\lemlab{L1:inequality}
	Consider any $k\times k$ matrix $A$, 	and $x\in \Re_+^k$, such that 
	$x$ is in the column space of $A$. Let 	$A^{\dagger}$ denote the 
	Moore-Penrose pseudo-inverse of $A$.	Then, 
$
	x^T A^{\dagger} x \geq \frac{\norm{1}{x}^2}{\mathrm{sum}(A)}
$.
\end{lemma}

\begin{proof}
	Consider any $x$ in the positive orthant and column space of $A$. 
	Let $v_1, \dots ,v_q$ be the eigenvectors of $A$ corresponding to it's 
	non-zero eigenvalues. 
	We may express $x$ in the form $x = \sum_i \beta_i v_i$, so that 
	\begin{align*}
		\norm{1}{x} = \langle \1, x \rangle = \sum_{i\in [q]} \beta_i \langle \1, 
		v_i \rangle
		\Rightarrow \norm{1}{x}^2 = (\sum_{i\in [q]} \beta_i \langle \1, v_i \rangle)^2.
	\end{align*}
	We also have 
	\[
		x^T A^{\dagger} x = x^T(\sum_{i\in [q]} \lambda_i^{-1} v_i v_i^T)x = 
		\sum_{i\in[q]} \lambda_i^{-1} \beta_i^2.
	\]
	Now by Cauchy-Schwartz, 
	\[
		\pth{\sum_i \lambda_i \langle \1, v_i \rangle^2 }\pth{\sum_{i\in [q]} 
		\lambda_i^{-1} \beta_i^2} 
		\geq \norm{1}{x}^2.
	\]
	Therefore, we have
	\[
		x^T A^{\dagger} x \geq \frac{||x||_1^2}{\sum_{i\in [q]} \lambda_i \langle \1, 
		v_i \rangle^2}
		= \frac{\norm{1}{x}^2}{\1^T A \1} 
		= \frac{\norm{1}{x}^2}{\mathrm{sum}(A)}.
	\]
\end{proof}
\medskip

\noindent
Equipped with all necessary tools, we may now prove our result.

\subsubsection{Our Gaussian Measure Bound}
Let $C\equiv (0,U,V)$ be a simplicial cone with apex at the origin. We now 
show an upper bound on the Gaussian measure of $C$ that depends surprisingly 
on only the para-volume and sum-norm of $U$. Since Gaussian measure is at most $1$, 
it is evident when viewing our bound that it can only be useful for simplicial 
cones wherein the sum-norm of their normal vectors is $O(\sqrt{k})$, and the 
para-volume of their normal vectors is not too small.
\begin{theorem}
\thmlab{cone:measure}
	Let $C\equiv (0,U,V)$ be a simplicial cone with apex at the origin. Let 
	$\ell = \norm{2}{\sum_i u_i}$ (i.e. sum-norm of the normal vectors), then 
	the Gaussian measure of $C$ is at most 
$
		\pth{\frac{e}{2\pi k}}^{k/2}\frac{\ell^{k}}{\sqrt{\detm{A_U}}}
$
\end{theorem}

\begin{proof}
By the sum-norm property, the sum of entries of $A_U$ is $\ell^2$. Also by the 
definition of a simplicial cone, $U$, and cosequently $A_U$, must have full rank. 
Thus we may apply \lemref{L1:inequality} over the entire positive orthant.
We proceed to analyze the multivariate integral in \Eqref{formulation}, by first 
applying \lemref{L1:inequality} and then linearizing the exponent using 
\lemref{linearize}. Post-linearization, our approach is similar to the 
presentation of Boeroeczky and Henk \cite{boroczky1999random}.
We have,
\begin{align*}
    I \leftarrow &\int\limits_{\Re^k_{+}} e^{-x^{T}A_U^{-1}x}\dif{x} \quad
    \leq \quad \int\limits_{\Re^k_{+}} e^{-\norm{1}{x}^2/\ell^2}\dif{x}
    \quad \text{(by \lemref{L1:inequality})}\quad
    = \quad \ell^{k}\int\limits_{\Re^k_{+}} e^{-\norm{1}{y}^2}\dif{y}
    \quad \text{(Subst. } y \leftarrow x/\ell) 
\end{align*}
\vspace{-1ex}
\begin{align*}
    =&~\frac{\ell^{k}}{\sqrt{\pi}}\int\limits_{\Re^k_{+}}\int\limits_{-\infty}^{\infty} e^{-~t^2~+~2it\sum\limits_{i\in [k]} y_i}\dif{t}\dif{y}
    \quad \text{(by \lemref{linearize})}\quad
    =\quad \frac{\ell^{k}}{\sqrt{\pi}}\int\limits_{-\infty}^{\infty} e^{-t^2}~\prod\limits_{i\in [k]}\pth{\int\limits_{0}^{\infty} e^{2ity_i}\dif{y_i}}\dif{t}
\end{align*}
\vspace{-1.3ex}
\begin{align*}
    =&~\frac{\ell^{k}}{\sqrt{\pi}}\int\limits_{-\infty}^{\infty} e^{-t^2} \pth{\int_{0}^{\infty} e^{2its} \dif{s}}^k \dif{t}
    &&\text{Let }g(t) = e^{-t^2} \pth{\int_{0}^{\infty} e^{2its} \dif{s}}^k \\
    ~&~ && \cardin{g(a+ib)} \leq e^{-a^2 + b^2} \pth{\int_{0}^{\infty} e^{-2bs} \dif{s}}^k \\
    ~&~ && \Rightarrow~ \lim\limits_{a\rightarrow \infty} g(a+ib) \rightarrow 0,~\forall b>0 \\
    =&~\frac{\ell^{k}}{\sqrt{\pi}}\int\limits_{-\infty}^{\infty} e^{-a^2 + b^2 -2abi} \pth{\int_{0}^{\infty} e^{-2bs + 2asi} \dif{s}}^k \dif{a}
    && \text{(for $b>0$)~~by \lemref{change:line}}\\
    =&~\frac{e^{k/2}~\ell^{k}}{\sqrt{\pi}(2k)^{k/2}}\int\limits_{-\infty}^{\infty} e^{-a^2} \pth{2be^{-ia/b}\int_{0}^{\infty} e^{-2bs + 2asi} \dif{s}}^k \dif{a}
    && \text{Fixing $b=\sqrt{k/2}$}\\
    = &~\frac{e^{k/2}~\ell^{k}}{\sqrt{\pi}(2k)^{k/2}}\cardin{\,\int\limits_{-\infty}^{\infty} e^{-a^2} \pth{2be^{-ia/b}\int_{0}^{\infty} e^{-2bs + 2asi} \dif{s}}^k \dif{a}\,}
    && \text{Since expr. is real and +ve}\\
    \leq &~\frac{e^{k/2}~\ell^{k}}{\sqrt{\pi}(2k)^{k/2}}\int\limits_{-\infty}^{\infty} e^{-a^2} \pth{2b\int_{0}^{\infty} e^{-2bs} \dif{s}}^k \dif{a}
    && \text{By \obsref{abs:val}}\\
    = &~\frac{e^{k/2}~\ell^{k}}{\sqrt{\pi}(2k)^{k/2}} \int\limits_{-\infty}^{\infty} e^{-a^2} \dif{a} ~=~ \frac{e^{k/2}~\ell^{k}}{(2k)^{k/2}}
\end{align*}
Lastly, the claim follows by substituting the above in \Eqref{formulation}.
\end{proof}

\subsection{Analysis of Hyperplane Rounding given Strong Colorability}
\seclab{analysis:strong}


In this section we analyze the performance of random hyperplane rounding 
on $k$-uniform hypergraphs that are $(k+\ell)$-strongly colorable.

\begin{theorem}
\thmlab{analysis:strong}
    Consider any $(k+\ell)$-strongly colorable $k$-uniform hypergraph 
    $H=(V,E)$. The expected fraction of monochromatic 
    edges obtained by performing random hyperplane rounding on the solution of 
    \rlxref{strong}, is 
$        O\pth{\,\ell^{k-1/2}\pth{\frac{e}{2\pi}}^{k/2}\frac{1}{k^{(k-1)/2}}\,}.
$
\end{theorem}

\begin{proof}
    Let $U$ be any $k\times k$ matrix whose columns are unit vectors 
    $u_1, \dots , u_k \in Re^k$ that satisfy the edge constraints in \rlxref{strong}. 
    Recall from \secref{setup:local}, that to bound the probability of a monochromatic 
    edge we need only bound the expression in \Eqref{monochr:prob} for $U$ of 
    the above form. 
    By \rlxref{strong}, the gram matrix $A_U = U^T U$, is exactly, 
$
        A_U = (1+\alpha)\mathrm{I} -\alpha\1\1^T$ where $
        \alpha = \frac{1}{k+\ell-1}$.
    By matrix determinant lemma (determinant formula for rank one updates), we know 
    \[
        \detm{A_U} = (1+\alpha)^{k}\pth{1-\frac{k\alpha}{1+\alpha}}
        \geq 
        \pth{\frac{\ell}{k+\ell}}
        = 
        \Omega\pth{\frac{\ell}{k}}
    \]
    Further, \rlxref{strong} implies the length of $\sum_i u_i$, is at most $\ell$. 
    The claim then follows by combining \Eqref{monochr:prob} 
    with \thmref{cone:measure}.
\end{proof}

\paragraph{Note.}
Being that any edge in the solution to the strong colorability relaxation corresponds 
to a symmetric cone, \thmref{analysis:strong} is directly implied by prior work on the 
volume of symmetric spherical simplices. It is in the next section, where the true 
power of \thmref{cone:measure} is realized.

\paragraph{Remark.}
As can be seen from the asymptotic volume formula of symmetric spherical simplices, 
$\sqrt{\pi k/(2e)}$ is a sharp threshold for $\ell$, i.e. when~
$\ell > (1+o(1))\sqrt{\pi k/(2e)}$, hyperplane rounding does worse than the naive 
random algorithm, and when~$\ell < (1-o(1))\sqrt{\pi k/(2e)}$, hyperplane rounding beats 
the naive random algorithm.

\subsection{Analysis of Hyperplane Rounding given Rainbow Colorability}
\seclab{analysis:rainbow}

In this section we analyze the performance of random hyperplane rounding 
on $k$-uniform hypergraphs that are $(k-\ell)$-rainbow colorable.

Let $U$ be the $k\times k$ matrix whose columns are unit vectors 
$u_1, \dots , u_k \in \Re^k$ satisfying the edge constraints in \rlxref{rainbow}. 
We need to bound the expression in \Eqref{monochr:prob} for $U$ of 
the above form. While we'd like to proceed just as in \secref{analysis:strong}, 
we are limited by the possibility of $U$ being singular or the parallelotope 
determined by $U$ having arbitrarily low volume (as $u_1$ 
can be chosen arbitrarily close to the span of $u_2,\dots ,u_k$ while still 
satisfying $||\sum_i u_i||_2 \leq \ell$). 
\medskip

While $U$ can be bad with respect to our properties of interest, we will show 
that some subset of the vectors in $U$ are reasonably well behaved with respect 
to para-volume and sum-norm. 

\subsubsection{Finding a Well Behaved Subset}
\seclab{wbs}

We'd like to find a subset of $U$ with high para-volume, or equivalently, 
a principal sub-matrix of $A_U$ with reasonably large determinant.
To this end, we express the gram matrix $A_U = U^TU$ as the sum of a 
symmetric skeleton matrix $B_U$ and a residue matrix $E_U$. 
Formally, $E_U = A_U - B_U$ and 
$B_U = (1+\beta)I - \beta\1\1^T$ where $\beta = \frac{1}{k-\ell-1}$.
We have (assuming $\ell = o(k)$),
$\mathrm{sum}(A_U) \leq~\ell^2$ and $\mathrm{sum}(B_U) =~k - k(k-1)\beta\quad =~ \frac{-\ell}{1-o(1)}$.
Let $s \leftarrow~\mathrm{sum}(E_U)\leq~\ell^2 - \mathrm{sum}(B_U) =~
\ell^2 + \frac{\ell}{1-o(1)}$.

We further observe that $E_U$ is symmetric, with all diagonal entries zero. 
Also since $u_1,\dots ,u_k$ satisfy \rlxref{rainbow}, all entries of $E_U$ 
are non-negative. 

By an averaging argument, at most $ck^{\delta}$ columns of $E_U$ have column sums 
greater than $s/(ck^{\delta})$ for some parameters $\delta,c$ to be determined later. 
Let $S\subseteq [k]$ be the set of indices of the columns having the lowest 
$k - ck^{\delta}$ column sums. Let $\tilde{k} \leftarrow |S| = k - ck^{\delta}$, 
and let $A_S, B_S, E_S$ be the corresponding matrices restricted to $S$ 
(in both columns and rows). 

\paragraph{Spectrum of $B_S$ and $E_S$.}

\begin{observation}
For a square matrix $X$, let $\eigmin{X}$ denote its minimum 
eigenvalue.
\obslab{B:spectrum}
    The eigenvalues of $B_S$ are exactly $(1+\beta)$ with multiplicity 
    $(\tilde{k}-1)$, and $(1+\beta -\tilde{k}\beta)$ with 
    multiplicity $1$. Thus $\eigmin{B_S}=1+\beta -\tilde{k}\beta$. 
    This is true since $B_S$ merely shifts all eigenvalues of $-\beta\1\1^T$ 
    by $1+\beta$.
\end{observation}

\noindent
While we don't know as much about the spectrum of $E_S$, we can still say 
some useful things.
\begin{observation}
\obslab{E:radius}
    Since $E_S$ is non-negative, by Perron-Frobenius theorem, its spectral radius 
    is equal to its max column sum, which is at most $s/(ck^{\delta})$. 
    Thus $\eigmin{E_S}\geq - s/(ck^{\delta})$.
\end{observation}

Now that we know some information about the spectra of $B_S$ and $E_S$, the next 
natural step is to consider the behaviour of spectra under matrix sums. 

\paragraph{Spectral properties of Matrix sums.}~\medskip

\noindent
The following identity is well known.
\begin{observation}
\obslab{sum:eigorder}
    If $X$ and $Y$ are symmetric matrices with eigenvalues 
    $x_1 > x_2 > \dots > x_m$ and $y_1 > y_2 > \dots > y_m$ and the 
    eigenvalues of $A+B$ are $z_1 > z_2 > \dots > z_m$, then 
    \[
        \forall~ 0\leq i+j \leq m,~~ z_{m-i-j} \geq x_{m-i} + y_{m-j}.
    \]
    In particular, this implies $\eigmin{X+Y} \geq \eigmin{X} + \eigmin{Y}$.
\end{observation}

\noindent
We may finally analyze the spectrum of $A_S$.

\paragraph{Properties of $A_S$.}
\begin{observation}[Para-Volume]
\obslab{A:spectrum}
    Let the eigenvalues of $A_S$ be $a_1 > a_2 > \dots > a_{\tilde{k}}$
    By \obsref{B:spectrum}, \obsref{E:radius}, and \obsref{sum:eigorder} we have 
    (Assuming $\ell < ck^{\delta/2}$), 
    \begin{alignat*}{2}
        \eigmin{A_S} = a_{\tilde{k}} \quad &\geq\quad 1+\beta -\tilde{k}\beta 
        -\frac{s}{ck^{\delta}} &&= 
        \quad \frac{c}{k^{1-\delta}}  - \frac{\ell^2}{ck^\delta} - o(1)
        \\
        a_2, a_3, \dots ,a_{\tilde{k}-1} \quad &\geq \quad 1+\beta -\frac{s}{ck^{\delta}}
        &&= \quad 1 - \frac{\ell^2}{ck^\delta} -o(1)
    \end{alignat*}
    Consequently,
    \[
        \qquad \detm{A_S} \geq  
        \pth{\frac{c}{k^{1-\delta}}  - \frac{\ell^2}{ck^\delta} - o(1)} 
        \pth{1 - \frac{\ell^2}{ck^\delta} - o(1)}^{\tilde{k}}
        \geq 
        \pth{\frac{c}{k^{1-\delta}}  - \frac{\ell^2}{ck^\delta} - o(1)}e^{-k}
    \]
    
    In particular, note that $A_S$ is non-singular and has non-negligible para-volume 
    when 
    \[
        \frac{\ell^2}{ck^\delta} = \frac{c}{2k^{1-\delta}}, \quad \text{i.e.~} \ell\approx 
        ck^{\delta - 1/2}
        \quad \text{or,~~} \delta\approx \frac{1}{2}\frac{\log(\ell/c)}{\log k}
    \]
\end{observation}    

\begin{observation}[Sum-Norm]
Since $E_U$ is non-negative, $\mathrm{sum}(E_S) \leq \mathrm{sum}(E_U) = s$.
Also we know that the sum of entries of $A_S$ is 
\begin{align}
\eqlab{A:sumnorm}
    \mathrm{sum}(B_S)+ \mathrm{sum}(E_S) &=~
    \tilde{k}(1+\beta) - \tilde{k}(\tilde{k} - 1)\beta + s 
    \leq ck^{\delta} + s
\end{align}
\end{observation}

\subsubsection{The Result.}
We are now equipped to prove our result.
\begin{theorem}
\thmlab{analysis:rainbow}
    For $\ell < \sqrt{k}/100$, consider any $(k-\ell)$-rainbow colorable $k$-uniform hypergraph 
    $H=(V,E)$. Let $\theta = 1/2 + \log(\ell)/\log(k)$ 
    and~~$\eta = 19(1-\theta)/40$. 
    The expected fraction of monochromatic 
    edges obtained by performing random hyperplane rounding on the solution of 
    \rlxref{rainbow}, is at most
    \[
        \frac{1}{2.1^{k}\,k^{\eta k}}
    \]
\end{theorem}

\begin{proof}
    Let $U$ be any $k\times k$ matrix whose columns are unit vectors 
    $u_1, \dots , u_k \in \Re^k$ that satisfy the edge constraints in \rlxref{strong}. 
    Recall from \secref{setup:local}, that to bound the probability of a monochromatic 
    edge we need only bound the expression in \Eqref{monochr:prob} for $U$ of 
    the above form. 
    
    By \secref{wbs}, we can always choose a matrix $U_S$ whose columns $\tilde{u}_1, 
    \dots ,\tilde{u}_{\tilde{k}}$ are from the set $\brc{u_1,\dots ,u_k}$, such that 
    the gram matrix $A_S = U_S^T U_S$ satisfies \Eqref{A:sumnorm} and \obsref{A:spectrum}. 
    Clearly the probability of all vectors in $U$ being monochromatic is at most the 
    probability of all vectors in $U_S$ being monochromatic. 
    
    Thus just as in \secref{setup:local}, to find the probability of $U_S$ being 
    monochromatic, we may assume without loss of generality that we are performing 
    random hyperplane rounding in $\Re^{\tilde{k}}$ on any $\tilde{k}$-dimensional vectors 
    $\tilde{u}_1, \dots, \tilde{u}_{\tilde{k}}$ whose gram (pairwise inner-product) matrix 
    is the aforementioned $A_S$.
    
    Specifically, by combining \Eqref{A:sumnorm} and \obsref{A:spectrum} with 
    \thmref{cone:measure}, our expression is at most:
    \begin{align*}
    		\pth{\frac{e}{2\pi}}^{\tilde{k}/2}\pth{\frac{ck^{\delta}+s}{k}}^{\tilde{k}/2}
    		\frac{1}{\sqrt{det(A_U)}}~
    		\leq ~
    		3.2^{\tilde{k}/2}\pth{\frac{(1-o(1))c}{k^{1-\delta}}}^{\tilde{k}/2}
    		\leq ~
    		\frac{1}{2.1^{k}\,k^{(1-c)(1-\delta)k}}
    \end{align*}
    assuming $c = 1/20$, $\delta\geq 1/2$ and $\ell < \sqrt{k}/100$ (constraint on $\ell$ ensures 
    that non-singularity conditions of \obsref{A:spectrum} are satisfied). 
    The claim follows. 
\end{proof}

\paragraph{Remark.}
Yet again we see a threshold for $\ell$, namely, when $\ell < \sqrt{k}/100$, hyperplane rounding  
beats the naive random algorithm, and for $\ell = \Omega(\sqrt{k})$, it fails to do better. 
In fact, as we'll see in the next section, assuming the UGC, we show a hardness result when 
$\ell = \Omega(\sqrt{k})$.

\section{Hardness of Max-$2$-Coloring under Low Discrepancy}
\seclab{hardness}
In this section we consider the hardness of Max-$2$-Coloring when promised 
discrepancy as low as one. As noted in \secref{analysis:rainbow}, our analysis 
requires the configuration of vectors in an edge to be well behaved with respect 
to sum-norm and para-volume. While in the discrepancy case, we can ensure good 
sum-norm, the vectors in an edge can have arbitrarily low para-volume. While in 
the rainbow case we can remedy this by finding a reasonably large well behaved 
subset of vectors, this is not possible in the case of discrepancy. 

Indeed, consider the following counterexample: Start in $2$ dimensions with 
$k/3$ copies each of any $u_1,u_2,u_3$ such that $u_1+u_2+u_3 = 0$. Lift all 
vectors to $3$-dimensions by assigning every vector a third coordinate of value 
exactly $1/k$. This satisfies \rlxref{discrepancy}, yet every superconstant 
sized subset has para-volume zero.

Confirming that this is not an artifact of our techniques and the problem is in fact hard, we show in this section via a reduction 
from Max-Cut, that assuming the Unique Games conjecture, it is NP-Hard to 
Max-$2$-Color much better than the naive random algorithm that miscolors $2^{-k + 1}$ fraction of edges, even in the case of 
discrepancy-$1$ hypergraphs. 

\subsection{Reduction from Max-Cut}
\label{subsec:maxcut}
Let $k = 2t + 1$. Let $G = (V, E)$ be an instance of Max-Cut, where each edge has 
weight 1. Let $n = |V|$ and $m = |E|$. We produce a hypergraph $H = (V', E')$ where 
$V' = V \times [k]$. For each $u \in V$, let $\cloud(u) := \{ u \} \times [k]$. For 
each edge $(u, v) \in E$, we add $N := 2 \binom{k}{t} \binom{k}{t + 1}$ hyperedges 
\[
\{U \cup V :  U \subseteq \cloud(u), V \subseteq \cloud(v), |U| + |V| = k, ||U| - |V|| = 1 \},
\]
each with weight $\frac{1}{N}$. Call these hyperedges {\em created by $(u, v)$}. 
The sum of weights is $m$ for both $G$ and $H$.

\subsubsection{Completeness}
Given a coloring $C : V \mapsto \{ B, W \}$ that cuts at least $(1 - \alpha)m$ edges of $G$, 
we color $H$ so that for every $v \in V$, each vertex in $\cloud(v)$ is given the same color as $v$. 
If $(u, v) \in E$ is cut, all hyperedges created by $(u, v)$ will have discrepancy 1. Therefore, 
the total weight of hyperedges with discrepancy 1 is at least $(1 - \alpha)m$.

\subsubsection{Soundness}
Given a coloring $C' : V' \mapsto \{ B, W \}$ such that the total weight of non-monochromatic 
hyperedges is $(1 - \beta)m$, $v \in V$ is given the color that appears the most in its cloud 
($k$ is odd, so it is well-defined). Consider $(u, v) \in E$. If no hyperedge created by 
$(u, v)$ is monochromatic, it means that $u$ and $v$ should be given different colors by the 
above majority algorithm (if they are given the same color, say white, then there are at least 
$t+1$ white vertices in both clouds, so we have at least one monochromatic hyperedge). 

This means that for each $(u, v) \in E$ that is uncut by the above algorithm (lost weight 1 
for Max-Cut objective), at least one hyperedge created by $(u, v)$ is monochromatic, and we lost 
weight at least $\frac{1}{N}$ there for our problem. This means that the total weight of cut 
edges for Max-Cut is at least $(1 - \beta N)m$.

\subsubsection{The Result}
\begin{theorem} [\cite{KKMO07}]
Let $G = (V, E)$ be a graph with $m = |E|$. 
For sufficiently small $\epsilon > 0$, it is UG-hard to distinguish the following cases.
\begin{itemize}
\item There is a 2-coloring that cuts at least $(1 - \epsilon)|E|$ edges.
\item Every 2-coloring cuts at most $(1 - (2 / \pi) \sqrt{\epsilon})|E|$ edges.
\end{itemize}
\end{theorem}
Our reduction shows that

\begin{theorem}
Given a hypergraph $H = (V, E)$, it is UG-hard to distinguish the following cases.
\begin{itemize}
\item There is a 2-coloring where at least $(1 - \epsilon)$ fraction of hyperedges have discrepancy 1.
\item Every 2-coloring cuts (in a standard sense) at most $(1 - (2 / \pi) \frac{\sqrt{\epsilon}}{N})$ fraction of hyperedges.
\end{itemize}
\end{theorem}
$N = 2\binom{k}{t}\binom{k}{t + 1} \leq (2 / \pi) 2^k \cdot 2^k \leq (2 / \pi)  2^{2k}$. If we take $\epsilon = 2^{-6k}$ for large enough $k$, we cannot distinguish 
\begin{itemize}
\item There is a 2-coloring where at least $(1 - 2^{-6k})$ fraction of hyperedges have discrepancy 1.
\item Every 2-coloring cuts (in a standard sense) at most $(1 - 2^{-5k})$ fraction of hyperedges.
\end{itemize}
This proves~\thmref{hardness1}.

\subsection{NP-Hardness}

In this subsection, we show that given a hypergraph which admits a 2-coloring with discrepancy at most 2, 
it is NP-hard to find a 2-coloring that has less than $k^{-O(k)}$ fraction of monochromatic hyperedges.
Note that while the inapproximability factor is worse than the previous subsection, we get NP-hardness and it holds when the input hypergraph is promised to have {\em all} hyperedges have discrepancy at most $2$.
The reduction and the analysis closely follow from the more general framework of Guruswami and Lee~\cite{GL15a} except that we prove a better reverse hypercontractivity bound for our case. 

\subsubsection{$Q$-Hypergraph Label Cover}
An instance of $Q$-Hypergraph Label Cover is based on a $Q$-uniform hypergraph $H = (V, E)$. Each hyperedge-vertex pair $(e, v)$ such that $v \in e$ is associated with a projection $\pi_{e, v} : [R] \rightarrow [L]$ for some positive integers $R$ and $L$. A labeling $l : V \rightarrow [R]$ {\em strongly satisfies} $e = \left\{v_1, \dots, v_Q \right\}$ when $\pi_{e, v_1}(l(v_1)) = \cdots = \pi_{e, v_Q}(l(v_Q))$. It {\em weakly satisfies} $e$ when $\pi_{e, v_i}(l(v_i)) = \pi_{e, v_j}(l(v_j))$ for some $i \neq j$. The following are two desired properties of instances of $Q$-Hypergraph Label Cover.
\begin{itemize}
\item Regular: every projection is $d$-to-$1$ for $d = R/L$.

\item Weakly dense: any subset of $V$ of measure at least $\epsilon$ vertices induces at least $\frac{\epsilon^Q}{2}$ fraction of hyperedges.

\item $T$-smooth: for all $v \in V$ and $i \neq j \in [R]$, 
$\quad \Pr_{e \in E : e \ni v} [\pi_{e, v}(i) = \pi_{e, v}(j)] \leq \frac{1}{T}$. 
\end{itemize}

The following theorem asserts that it is NP-hard to find a good labeling in such instances. 

\begin{theorem}[\cite{GL15a}]
For all integers $T, Q \geq 2$ and $\eta > 0$, the following is true. Given an instance of $Q$-Hypergraph Label Cover that is regular, weakly-dense and $T$-smooth, it is NP-hard to distinguish between the following cases.
\begin{itemize}
\item Completeness: There exists a labeling $l$ that strongly satisfies every hyperedge.
\item Soundness: No labeling $l$ can weakly satisfy $\eta$ fraction of hyperedges.
\end{itemize}
\label{thm:PCP}
\end{theorem}

\subsubsection{Distributions}
We first define the distribution $\overline{\mu}'$ for each {\em block}. 
$2Q$ points $x_{q, i} \in \{ 1, 2 \}^d$ for $1 \leq q \leq Q$ and $1 \leq i \leq 2$ are sampled by the following procedure.

\begin{itemize}
\item Sample $q' \in [Q]$ uniformly at random.
\item Sample $x_{q', 1}, x_{q', 2} \in \{ 1, 2 \}^d$ i.i.d.
\item For $q \neq q'$, $1 \leq j \leq d$, sample a permutation $((x_{q, 1})_j, (x_{q, 2})_j) \in \{ (1, 2), (2, 1) \}$ uniformly at random. 
\end{itemize}

\subsubsection{Reduction and Completeness}

We now describe the reduction from $Q$-Hypergraph Label Cover. Given a $Q$-uniform hypergraph $H = (V, E)$ with $Q$ projections from $[R]$ to $[L]$ for each hyperedge (let $d = R/L$), the resulting instance of $2Q$-Hypergraph Coloring is $H' = (V', E')$ where $V' = V \times \{ 1, 2 \}^R$. Let $\cloud(v) := \left\{ v \right\} \times \{ 1, 2 \}^R$. 
The set $E'$ consists of hyperedges generated by the following procedure.

\begin{itemize}
\item Sample a random hyperedge $e = (v_1, \dots, v_Q) \in E$ with associated projections $\pi_{e, v_1}, \dots, \pi_{e, v_Q}$ from $E$. 

\item Sample $( x_{q, i} )_{1 \leq q \leq Q, 1 \leq i \leq 2} \in \{ 1, 2 \}^{R}$ in the following way. For each $1 \leq j \leq L$, independently sample 
$( (x_{q, i})_{\pi_{e, v_q}^{-1}(j)} )_{q, i}$ from $((\{ 1, 2 \}^d){2Q}, \overline{\mu}')$. 

\item Add a hyperedge between $2Q$ vertices $\left\{ (v_q, x_{q, i}) \right\}_{q, i}$ to $E'$. We say this hyperedge is {\em formed from} $e \in E$. 
\end{itemize}

Given the reduction, completeness is easy to show.

\begin{lemma}
If an instance of $Q$-Hypergraph Label Cover admits a labeling that strongly satisfies every hyperedge $e \in E$, there is a coloring $c : V' \rightarrow \{ 1, 2 \}$ of the vertices of $H'$ such that every hyperedge $e' \in E'$ has at least $(Q - 1)$ vertices of each color. 
\end{lemma}
\begin{proof}
Let $l : V \rightarrow [R]$ be a labeling that strongly satisfies every hyperedge $e \in E$. 
For any $v \in V, x \in \{ 1, 2 \}^R$, let $c(v, x) = x_{l(v)}$.
For any hyperedge $e' = \left\{ (v_q, x_{q, i}) \right\}_{q, i} \in E'$, 
$c(v_q, x_{q, i}) = (x_{q, i})_{l(v_q)}$,
and all but one $q$ satisfies $\left\{ (x_{q, 1})_{l(v_q)}, (x_{q, 2})_{l(v_q)} \right\} = \{ 1, 2 \}$. 
Therefore, the above strategy ensures that every hyperedge of $E'$ contains at least $(Q-1)$ vertices of each color.
\end{proof}

\subsubsection{Soundness}
\label{subsec:soundness}
\begin{lemma}
\label{lem:soundness-main}
There exists $\eta := \eta(Q)$ such that 
if $I \subseteq V'$ of measure $\frac{1}{2}$ induces less than $Q^{-O(Q)}$ fraction of hyperedges in $H'$, 
the corresponding instance of $Q$-Hypergraph Label Cover admits a labeling that weakly satisfies a fraction $\eta$ of hyperedges.
\end{lemma}
\begin{proof}
Consider a vertex $v$ and hyperedge $e \in E$ that contains $v$ with a permutation $\pi = \pi_{e, v}$. 
Let $f : \{ 1, 2 \}^R \mapsto [ 0 , 1 ]$ be a {\em noised} indicator function of $I \cap \cloud(v)$
with $\E_{x \in \{ 1, 2 \}^R}[f(x)] \geq \frac{1}{2} - \epsilon$ for small $\epsilon > 0$ that will be determined later. We define the inner product
\[
\langle f, g \rangle = \E_{x \in \{ 1, 2 \}^R} [f(x)g(x)].
\]
$f$ admits the Fourier expansion
\[
\sum_{S \subseteq [R]} \hat{f}(S) \chi_S
\]
where 
\[
\chi_S(x_1, \dots, x_k) = \prod_{i \in S} (-1)^{x_i}, \qquad 
\hat{f}(S) = \langle f , \chi_S \rangle.
\]
In particular, $\hat{f}(\emptyset) = \E[f(x)]$, and 
\begin{equation}
\label{eq:fourier}
\sum_{S} \hat{f}(S)^2 = \E[f(x)^2] \leq \E[f(x)]\
\end{equation}
A subset $S \subseteq [R]$ is said to be {\em shattered} by $\pi$ if $|S| = |\pi(S)|$. For a positive integer $J$, we decompose $f$ as the following:
\begin{align*}
f^{\good} &= \sum_{S : \mbox{ shattered}} \hat{f}(S) \chi_S \\
f^{\bad} &= f - f^{\good}.
\end{align*}
By adding a suitable noise and using smoothness of Label Cover, for any $\delta > 0$, we can assume that $||f^{\bad}||^2 \leq \delta$. 
See~\cite{GL15a} for the details.

Each time a $2Q$-hyperedge is sampled is formed from $e$, two points are sampled from each cloud. Let $x, y$ be the points in $\cloud(v)$. 
Recall that they are sampled such that for each $1 \leq j \leq L$,
\begin{itemize}
\item With probability $\frac{1}{Q}$, for each $i \in \pi^{-1}(j)$, $x_{i}$ and $y_{i}$ are independently sampled from $\{ 1, 2 \}$.
\item With probability $\frac{Q - 1}{Q}$, for each $i \in \pi^{-1}(j)$, $(x_{i}, y_{i})$ are sampled from $\{ (1, 2), (2, 1) \}$.
\end{itemize}

We can deduce the following simple properties.
\begin{enumerate}
\item $\E_{x, y}[\chi_{\{ i \}}(x) \chi_{\{ i \}}(y)] = - \frac{Q - 1}{Q}$. Let $\rho := - \frac{Q - 1}{Q}$.
\item $\E_{x, y}[\chi_{\{ i \}}(x) \chi_{\{ j \}}(y)] = 0$ if $i \neq j$.
\item $\E_{x, y}[\chi_{S}(x) \chi_{T}(y)] = 0$ unless $\pi(S) = \pi(T) = \pi(S \cap T)$. 
\end{enumerate}

We are interested in lower bounding 
\[
\E_{x, y} [f(x) f(y)] \geq \E[f^{\good}(x) f^{\good}(y)] 
- 3 \| f^{\bad}(x) \| ^2 \| f \|^2 \geq \E[f^{\good}(x) f^{\good}(y)] 
- 3 \delta.
\]

By the property 3., 
\begin{align*}
\E[f^{\good}(x) f^{\good}(y)] 
&= \sum_{S : \mbox{ shattered}} \hat{f}(S)^2 \rho^{|S|} \\
& = \E[f]^2 + 
\sum_{S : \mbox{ shattered}} \hat{f}(S)^2 \rho^{|S|} \\
& \geq \E[f]^2 + \rho (\sum_{|S| > 1} \hat{f}(S)^2) \qquad \mbox { since } \rho \mbox { is negative}\\
& \geq \E[f]^2 + \rho (\E[f] - \E[f]^2) \quad \mbox{ by~\eqref{eq:fourier}}  \\
& \geq \E[f]^2 (1 + \rho) - \epsilon \quad \mbox{ since } \E[f] \geq \frac{1}{2} - \epsilon \Rightarrow \E[f] - \E[f]^2 \leq \E[f]^2 + \epsilon \\
& \geq \frac{\E[f]^2}{Q} - \epsilon.
\end{align*}

By taking $\epsilon$ and $\delta$ small enough, we can ensure that 
\begin{equation}
\label{eq:reverse}
\E[f(x)f(y)] \geq \zeta := \frac{1}{5Q}.
\end{equation}
The soundness analysis of Guruswami and Lee~\cite{GL15a} ensures (\eqref{eq:reverse} replaces their Step 2) that there exists $\eta := \eta(Q)$ such that if 
the fraction of hyperedges induced by $I$ is less than $Q^{-O(Q)}$, the Hypergraph Label Cover instance admits a solution that satisfies $\eta$ fraction of constraints. We omit the details.
\end{proof}

\subsubsection{Corollary to Max-2-Coloring under discrepancy $O(\log k)$}
The above NP-hardness, 
combined with the reduction techinque from Max-Cut in Section~\ref{subsec:maxcut}, 
shows that
given a $k$-uniform hypergraph, it is NP-hard to distinguish whether it has discrepancy at most $O(\log k)$ or any $2$-coloring leaves at least $2^{-O(k)}$ fraction of hyperedges monochromatic.
Even though the direction reduction from Max-Cut results in a similar inapproximability factor with discrepancy even $1$, 
this result does not rely on the UGC and hold even {\em all} edges (compared to {\em almost} in Section~\ref{subsec:maxcut}) have discrepancy $O(\log k)$. 

Let $r = \Theta(\frac{k}{\log k})$ so that $s = \frac{k}{r} = \Theta(\log k)$ is an integer. 
Given a $r$-uniform hypergraph, it is NP-hard to distinguish whether it has discrepancy at most $2$ or any $2$-coloring leaves at least $r^{-O(r)}$ fraction of hyperedges monochromatic.
Given a $r$-uniform hypergraph, the reduction replaces each vertex $v$ with $\cloud(v)$ that contains $(2s - 1)$ new vertices. 
Each hyperedge $(v_1, \dots, v_r)$ is replaced by $d := (\binom{2s - 1}{s})^r \leq (2^s)^r = 2^k$ hyperedges
\[
\{ \cup_{i = 1}^r {V_i} : V_i \subset \cloud(v_i), |V_i| = s \}.
\]
If the given $r$-uniform hypergraph has discrepancy at most $2$, the resulting $k$-uniform hypergraph has discrepancy at most $2s = O(\log k)$.

If the resulting $k$-uniform hypergraph admits a coloring that leaves $\alpha$ fraction of hyperedges monochromatic, 
giving $v$ the color that appears more in $\cloud(v)$ is guaranteed to leaves at most $d \alpha$ fraction of hyperedges monochromatic.
Therefore, if any $2$-coloring of the input $r$-uniform hypergraph leaves at least $r^{-O(r)}$ fraction of hyperedges monochromatic, 
any $2$-coloring of the resulting $k$-uniform hypergraph leaves at least $\frac{r^{-O(r)}}{d} = 2^{-O(k)}$ fraction of hyperedges.

\subsection{Hardness of Max-2-Coloring under almost $(k-\sqrt{k})$-colorability}
Let $k$ be such that $\ell := \sqrt{k}$ be an integer and let $\chi := k - \ell$.
We prove the following hardness result for any $\epsilon > 0$ assuming the Unique Games Conjecture:
given a $k$-uniform hypergraph such that there is a $\chi$-coloring that have at least $(1 - \epsilon)$ fraction of hyperedges rainbow,
it is NP-hard to find a $2$-coloring that leaves at most $(\frac{1}{2})^{k - 1}$ fraction of hyperedges monochromatic.

The main technique for this result is to show the existence of a {\em balanced pairwise independence distribution} with the desired support. 
Let $\mu$ be a distribution on $[\chi]^k$. $\mu$ is called balanced pairwise independent if for any $i \neq j \in [k]$ and $a, b \in [\chi]$, 
\[
\Pr_{(x_1, \dots, x_k) \sim \mu} [x_i = a, x_j = b] = \frac{1}{\chi^2}.
\]
For example, the uniform distribution on $[\chi]^k$ is a balanced pairwise distribution. We now consider the following distribution $\mu$ to sample $(x_1, \dots, x_k) \in [\chi]^k$. 
\begin{itemize}
\item Sample $S \subseteq [k]$ with $|S| = \chi$ uniformly at random. Let $S = \{ s_1 < \dots < s_{\chi } \}$.
\item Sample a permutation $\pi : [\chi] \mapsto [\chi]$.
\item Sample $y \in [\chi]$.
\item For each $i \in [k]$, if $i = s_j$ for some $j \in [\chi]$, output $x_i = \pi(\chi)$. Otherwise, output $x_i = y$.
\end{itemize}
Note that for any supported by $(x_1, \dots, x_k)$, we have $\{ x_1, \dots, x_k \} = [\chi]$. 
Therefore, $\mu$ is supported on {\em rainbow strings}. 
We now verify pairwise independence. 
Fix $i \neq j \in [k]$ and $a, b \in [\chi]$. 
\begin{itemize}
\item If $a = b$, by conditioning on wheter $i, j$ are in $S$ or not, 
\begin{align*}
\Pr_{\mu} [x_i = a, x_j = b] = & \Pr[x_i = a, x_j = b | i, j \in S] \Pr[i, j \in S] +  \\
& \Pr[x_i = a, x_j = b | i \in S, j \notin S] \Pr[i \in S, j \notin S] + \\
& \Pr[x_i = a, x_j = b | i \notin, j \in S] \Pr[i \notin, j \in S] + \\
& \Pr[x_i = a, x_j = b | i, j \notin S] \Pr[i, j \notin S] \\
= & 0 \cdot (\frac{\chi(\chi - 1)}{k(k - 1)}) + 2\cdot (\frac{1}{\chi^2}) \cdot (\frac{l \chi}{k(k-1)}) + (\frac{1}{\chi}) \cdot (\frac{\ell(\ell-1)}{k(k - 1)}) \\
= & \frac{2 \ell \chi + \chi (\ell^2 - \ell)}{\chi^2 k(k-1)}
= \frac{\chi k + \chi \sqrt{k}}{\chi^2 k(k-1)}
= \frac{\sqrt{k}(\sqrt{k} + 1)}{\chi k(\sqrt{k} + 1)(\sqrt{k} - 1)} \\
= & \frac{1}{\chi (k - \sqrt{k})} = \frac{1}{\chi^2}.
\end{align*}
\item If $a \neq b$, by the same conditioning, 
\begin{align*}
\Pr_{\mu} [x_i = a, x_j = b] =& (\frac{1}{\chi(\chi - 1)}) \cdot (\frac{\chi(\chi - 1)}{k(k - 1)}) + 2\cdot (\frac{1}{\chi^2}) \cdot (\frac{\ell \chi}{k(k-1)}) + 0 \cdot(\frac{\ell(\ell-1)}{k(k - 1)}) \\
= & \frac{\chi^2 + 2l\chi}{\chi^2 k(k-1)}
= \frac{\chi + 2\ell}{\chi k(k-1)}
= \frac{k + \sqrt{k}}{\chi k(k-1)}
= \frac{1}{\chi^2}.
\end{align*}
\end{itemize}
Given such a balanced pairwise independent distribution supported on rainbow strings, a standard procedure following the work of Austrin and Mossel~\cite{AM09} shows that it is UG-hard to outperform the random $2$-coloring. We omit the details.

\section{Approximate Min-Coloring}
In this section, we provide approximation algorithms for the Min-Coloring 
problem under strong colorability, rainbow colorability, and low discrepancy 
assumptions. Our approach is standard, namely, we first apply degree 
reduction algorithms followed by the usual paradigm pioneered by Karger, Motwani 
and Sudan \cite{karger1998approximate}, for coloring bounded degree (hyper)graphs. Consequently, our 
exposition will be brief and non-linear. 
\medskip

In the interest of clarity, all results henceforth assume the special cases of 
Discrepancy $1$, or $(k-1)$-rainbow colorability, or $(k+1)$-strong colorability. 
All arguments generalize easily to the cases parameterized by $l$.

\subsection{Approximate Min-Coloring in Bounded Degree Hypergraphs}
\subsubsection{The Algorithm}

INPUT: $k$-uniform hypergraph $H=([n],E)$ with max-degree $t$ and $m$ edges, having 
Discrepancy $1$, or being $(k-1)$-rainbow colorable, or being $(k+1)$-strong colorable.
\bigskip

\begin{compactenum}
    \item Let $u_1,\dots ,u_n$ be a solution to the SDP relaxation 
    from \secref{relaxations} corresponding to the assumption on the hypergraph.
    
    \medskip
    
    \item Let $H_1$ be a copy of $H$, and let $\gamma,\tau$ be parameters to be 
    determined shortly.
    
    \medskip    
    
    \item Until no vertex remains in the hypergraph, Repeat:
    \begin{compactenum}[\qquad]

    \item Find an independent set $\mathcal{I}$ in the residual hypergraph, of size at 
    least $\gamma n$ by repeating the below process until $\cardin{\mathcal{I}}
    \geq \gamma n$:
    \begin{compactenum}[\qquad(A)]
        \item Pick a random vector $r$ from the standard multivariate normal 
        distribution. 
        
        \item For all $i$, if $\iprod{u_i}{r}\geq \tau$, add vertex $i$ to $\mathcal{I}$.
        
        \item For every edge $e$ completely contained in $\mathcal{I}$, delete any 
        single vertex in $e$, from $\mathcal{I}$.
    \end{compactenum}
    \medskip
    
    \item Color $\mathcal{I}$ with a new color and remove $\mathcal{I}$ and all edges 
    involving vertices in $\mathcal{I}$, from $H_1$.
    \end{compactenum}
\end{compactenum}

\subsubsection{Analysis}
\seclab{KMS:analysis}
First note that by \lemref{gauss:affine:equiv}, for any fixed vector $a$, $\iprod{a}{r}$ 
has the distribution $\normal{0}{1}$. 
Note that all SDP formulations in \secref{relaxations} satisfy,
\begin{equation}
\eqlab{sum-norm}
\norm{2}{\sum\limits_{j\in [k]} u_{i_j}}\leq 1
\end{equation}
Now consider any edge $e=(i_1,\dots , i_k)$. 
In any fixed iteration of the inner loop, the probability of $e$ being contained in 
$\mathcal{I}$ at Step (B), is at most the probability of 
\[
\iprod{r}{\sum\limits_{j\in [k]} u_{i_j}} \geq k\tau
\]
However, by \lemref{gauss:affine:equiv} and \Eqref{sum-norm}, the inner product above is 
dominated by the distribution $\normal{0}{1}$. Thus in any fixed iteration of the inner 
loop, let $H_1$ have $n_1$ vertices and $m_1$ edges, we have
\begin{align*}
    \Ex{\mathcal{I}} &\geq n_1 \Phi(\tau) - m_1\Phi(k\tau)\\
    &\geq n_1 e^{-\tau^2/2} - \frac{n_1 t}{k} e^{-k^2\tau^2/2}\\
    &= \Omega\pth{\gamma n_1}
    && \text {setting, }\tau^2 = \frac{2\log t}{k^2-1},\text{~~and~~} 
    \gamma = t^{-1/(k^2-1)}
\end{align*}
Now by applying Markov's inequality to the vertices not in $\mathcal{I}$, we have, 
$\Prob{\cardin{\mathcal{I}}< \gamma n_1}\leq 1-\Omega(\gamma)$. 
Thus for a fixed iteration of the outer loop, with high probability, the inner loop doesn't 
repeat more than $O(\log n_1/\gamma)$ times. 

Lastly, the outermost loop repeats $O(\log n/\gamma)$ times, using one color at each 
iteration. Thus with high probability, in polynomial time, the algorithm colors $H$ with 
\[
    t^{\frac{1}{k^2-1}}\log{n}~ \text{ colors.}
\]
\paragraph{Important Note.}
We can be more careful in the above analysis for the rainbow and strong colorability 
cases. Specifically, the crux boils down to finding the gaussian measure of the cone 
given by $\brc{x\sep{U^Tx\geq \tau}}$ instead of zero. Indeed, on closely following the 
proof of \thmref{cone:measure} we obtain for strong and rainbow coloring respectively (assuming max-degree $n^{k}$), 
\[
    n^{\frac{1}{k}\pth{1-\frac{3\beta}{2}}}\log n \qquad\text{and}\qquad
    n^{\frac{1}{k}\pth{1-\frac{5\beta}{4}}}\log n, \qquad\text{where~~}
    \beta = \frac{\log k}{\log n}
\]
While these improvements are negligible for small $k$, they are significant when 
$k$ is reasonably large with respect to $n$.
\bigskip

\subsection{Main Min-Coloring Result}
Combining results from \secref{KMS:analysis} with our degree reduction approximation schemes from the forthcoming sections, we obtain the following.
\begin{theorem}
\thmlab{min:coloring}
    Consider any $k$-uniform hypergraph $H=(V,E)$ with $n$ vertices. In $n^{c+O(1)}$ time, 
    one can color $H$ with 
    \begin{align*}
        &\min\brc{\pth{\frac{n}{c\log n}}^{\alpha},~
         n^{\frac{1}{k}\pth{1-\frac{3\beta}{2}}}, 
         \pth{\frac{m}{n}}^{\frac{1}{k^2}}}\log n
        \text{ colors,} &&\text{if $H$ is $(k+1)$-strongly colorable.}
        \\\\        
        &\min\brc{\pth{\frac{n}{c}}^{\alpha},~
         n^{\frac{1}{k}\pth{1-\frac{5\beta}{4}}}, 
         \pth{\frac{m}{n}}^{\frac{1}{k^2}}}\log n
        \text{ colors,} &&\text{if $H$ is $(k-1)$-rainbow colorable.}
        \\\\      
        &\min\brc{
        \pth{\frac{n}{c}}^{\alpha},~\pth{\frac{m}{n}}^{\frac{1}{k^2}}
        }\log n
        \text{~colors,} &&\text{if $H$ has discrepancy $1$.}\\
        &\text{where,\quad } \alpha = \frac{1}{k+2-o(1)}, \quad \beta = \frac{\log k}{\log n}
    \end{align*}
\end{theorem} 

\paragraph{Remark.}
In all three promise cases the general polytime min-coloring guarantee parameterized 
by $\ell$, is roughly $n^{\ell^2/k}$. Thus, the threshold value of $\ell$, for which 
standard min-coloring techniques improve with $k$, is $o(\sqrt{k})$.

\paragraph{Degree Reduction Schemes under Promise.}
Wigderson \cite{wigderson1983improving} and Alon et~al. \cite{AKMH96} studied degree reduction 
in the cases of $3$-colorable graphs and $2$-colorable hypergraphs, 
respectively.
Assuming our proposed structures, we are able to combine some simple 
combinatorial ideas with counterparts of the 
observations made by Wigderson and Alon et~al., to obtain degree 
reduction approximation schemes. Such approximation schemes 
are likely not possible assuming only $2$-colorability.
\subsection{Degree Redution under strong colorability}
\seclab{scdr}

Let $H=\pth{V,E\subseteq {V \choose k}}$ be a $k$-uniform 
$(k+1)$-strongly colorable hypergraph with $n$ vertices and $m$ 
edges. In this section, 
we give an algorithm that in $n^{c+O(1)}$ time, partially 
colors $H$ with $3n(k+1)\log k/(t^{1/(k-1)}c\log n)$ colors, such 
that no edge in the colored subgraph is monochromatic, 
and furthermore, the subgraph induced by the the uncolored 
vertices has max-degree $t$.

\bigskip

\noindent
The following observations motivate the structure of our algorithm. 

\begin{observation}
\obslab{s:wigderson:step}
    For any $(k+1)$-strong coloring $f: V \mapsto \pbrcx{k+1}$, of a $k$-uniform hypergraph $H$, and
    any subset of vertices $\overline{V}$ satisfying, $~\forall u,v\in \overline{V}, f(u)=f(v)=j$ 
    (all of the same color), the subgraph $F$ of $H$, induced by $N(\overline{V})$, is $k$-uniform and 
    $k$-strongly colorable. This is because $f$ is a strong coloring of $F$, and moreover, 
    $\forall\, v\in N(\overline{V}),~f(v)\neq j$, since $v$ has a neighbor in $\overline{V}$
    with color $j$. Thus we can $2$-color such a subgraph $F$ in polynomial time. 
\end{observation}
\noindent

\begin{observation}
\obslab{s:brute:dom:set}
By \obsref{s:wigderson:step}, in order to $3(k+1)$-color the subgraph induced by 
$\overline{V}\cup N(\overline{V})$ for an arbitrary subset $\overline{V}$ of vertices, 
we need only search through all possible $(k+1)$-colorings of $\overline{V}$, and then 
attempt to $2$-color the neighborhood of each color class with two new colors. This 
process will always terminate with some proper coloring of $\overline{V}\cup 
N(\overline{V})$.
\end{observation}

\bigskip

\noindent
We are now prepared to state the algorithm.

\subsubsection{The Algorithm \emph{SCDegreeReduce}}
\alglab{SCDegreeReduce}

\begin{compactenum}
    \item Let $H_1$ be a copy of $H$.
    
    \medskip
    
    \item \textbf{While} $H_1$ contains a vertex of degree greater than $t$:
    \medskip
    \begin{compactenum}[\qquad(A)]
        \item Let $H_2$ be a copy of $H_1$.
        
        \medskip
    
        \item Sequentially pick arbitrary vertices $\overline{V}=\brc{v_1, v_2 \dots v_s}$ 
        of degree at least $t$ from $H_2$, wherein we remove from $H_2$ 
        the vertices $\brc{v_i}\cup N(v_i)$ and all involved edges, after 
        picking $v_i$ and before picking $v_{i+1}$. We only stop when we 
        have either picked $c\log n/\log k$ vertices, or $H_2$ has max-degree $t$.
        
        \bigskip
        
        \item For every possible assignment of $k+1$ new colors 
        $\brc{c_1, \dots c_{k+1}}$ to the vertices in $\overline{V}$:
        
        \medskip
        
        \begin{compactenum}[\qquad(C1)]
            \item Let $C_i=\brc{u \sep{v\in \overline{V},~color(v)=c_i,~u\in N_{H_1}(v)}}$. Then 
            for each $i\in \pbrcx{k+1}$, $2$-color the subgraph of $H_1$ induced by $N_{H_1}(C_i)$ 
            using two new colors and the proper $2$-coloring algorithm for $r$-uniform, $r$-strongly colorable 
            graphs. 

            \medskip
            
            \item \textbf{If} no edge is monochromatic: 
            
            \begin{compactenum}[\qquad]
                \item Stick with this $3(k+1)$-coloring of $\overline{V}\cup N_{H_1}(\overline{V})$, 
                remove $\overline{V}\cup N_{H_1}(\overline{V})$ and all edges containing any of these 
                vertices, from $H_1$, and stop iterating through assignments of $\overline{V}$.
            \end{compactenum}
            
            \medskip            
            
            \item \textbf{If} some edge is monochromatic: 
            
            \begin{compactenum}[\qquad]
                \item Discard the coloring and continue iterating through assignments of 
                $\overline{V}$.
            \end{compactenum}
        \end{compactenum}
    \end{compactenum}
    
    \item[] \textbf{End~While}
    
    \bigskip    
    
    \item \textbf{Output} the partial coloring of $H$ and the residual graph $H_1$ of max-degree $t$.
\end{compactenum}

\subsubsection{The Result}

\begin{theorem}
\thmlab{scdr}
    Let $H=\pth{V,E\subseteq {V \choose k}}$ be a $k$-uniform $(k+1)$-strongly colorable hypergraph with 
    $n$ vertices. \algref{SCDegreeReduce} partially colors $H$ in $n^{c+O(1)}$ time, with at most 
    $\frac{3n(k+1)\log k}{t^{1/(k-1)}c\log n}$ colors, such that:
    \medskip
    
    \begin{compactenum}
        \item The subgraph of $H$ induced by the colored vertices has no monochromatic vertices.
        
        \medskip
        
        \item The subgraph of $H$ induced by the uncolored vertices has maximum degree $t$.
    \end{compactenum}
\end{theorem}

\begin{proof}
\obsref{s:wigderson:step} combined with the fact that step (C1) uses two new colors for each $C_i$, 
establishes that step (C) of \algref{SCDegreeReduce} will always terminate with 
some proper coloring of $\overline{V}\cup N_{H_1}(\overline{V})$. Furthermore, any edge intersecting 
$\overline{V_1}\cup N_{H_1}(\overline{V_1})$ and $\overline{V_2}\cup N_{H_1}(\overline{V_2})$ for $V_1$ 
and $V_2$ taken from different iterations of \algref{SCDegreeReduce}, cannot be monochromatic since we 
use new colors in each iteration. Thus the partial coloring is proper.
\medskip

\noindent
For the claim on number of colors, observe that a vertex of degree at least $t$, must have at least 
$(k-1)t^{1/(k-1)}$ distinct neighbors. Thus step (C) can be run at most $n/t^{1/(k-1)}$ times, using 
$3(k+1)$ new colors each time.

\noindent
Lastly for the runtime, note that for each run of step (C), there are at most $(k+1)^{c\log n/\log k} = 
n^{c+O(1)}$ assignments to try, and the rest of the work takes $n^{O(1)}$ time.
\end{proof}

\bigskip

\paragraph{Remark.}
We contrast \thmref{scdr} with the results of Alon et~al. \cite{AKMH96}, who give a polynomial time 
algorithm for degree reduction in $2$-colorable $k$-uniform hypergraphs using $O(n/t^{1/(k-1)})$ colors. 
The strong coloring property, gives us additional power, namely,  we obtain an approximation scheme, 
and furthermore, for constant $c$, \thmref{scdr} uses fewer colors than the result of Alon et~al.,
by a factor of about $k\log n/\log k$.

\noindent
The arguments in this section and the next are readily generalizable - One can modify the degree 
reduction algorithm, such that the bound on colors used, would be a function of the strong colorability parameter of the hypergraph.

\subsection{Degree Reduction under Low Discrepancy}
\seclab{lddr}

For odd $k$, let $H=\pth{V,E}$ be a $k$-uniform hypergraph with 
$n$ vertices, that admits a discrepancy $1$ coloring. In this section, 
we give an algorithm that in $n^{c+O(1)}$ time, partially 
colors $H$ with $3n(k+1)/(t^{\frac{1}{k-1}}c\log n)$ colors, such 
that no edge in the induced colored subgraph is monochromatic, 
and furthermore, the subgraph induced by the the uncolored 
vertices has max-degree $t$.
\bigskip

\noindent
First, we present a warmup algorithm that exposes the key ideas.
The following observations motivate the structure of our algorithm. 

\begin{observation}
\obslab{d:wigderson:step}
    For any discrepancy $1$ coloring $f: V \mapsto \brc{-1,1}$, of a $k$-uniform 
    hypergraph $H$, and any size $k-1$ subset of vertices $S$, we have:
    \medskip
    \begin{compactenum}[(A)]
        \item If $N(S)$ is an independent set, we can properly $2$-color the subgraph induced by $S \cup N(S)$.
        
        \bigskip
        
        \item If $N(S)$ contains an edge, then the set $S$ has discrepancy $0$ in 
        the coloring $f$. This is because, an edge cannot be monochromatic in the 
        coloring $f$, and by assumption, $S$ must be have a neighbor 
        with color $-1$ and a neighbor with color $+1$. 
    \end{compactenum}
\end{observation}

Though \obsref{d:wigderson:step} and \obsref{s:wigderson:step} are functionally 
similar, the two-pronged nature of \obsref{d:wigderson:step} almost wholly accounts 
for the gap in power between the respective degree reduction algorithms. Intuitively, 
the primary weakness comes from the fact that $N(S)$ being an independent set tells 
us nothing about the discrepancy of $S$.
\bigskip

\noindent
Nevertheless, we may still exploit some aspects of this observation.
\begin{observation}
\obslab{d:marked}
    Consider any discrepancy $1$ coloring $f: V \mapsto \brc{-1,1}$, of a 
    $k$-uniform hypergraph $H$, and any set of subsets $S_1,\dots S_m$ each of size 
    $(k-1)$ and discrepancy $0$ in the coloring $f$. The $(k-1)$-uniform hypergraph 
    $F$ with vertex set $\bigcup_i S_i$ and edge set $\brc{S_1,\dots S_m}$, has a 
    discrepancy $0$ coloring ($f$). Thus we can properly $2$-color $F$ in polynomial time. 
\end{observation}

\noindent
We are now ready to state the warmup algorithm, whose correctness is evident from 
\obsref{d:wigderson:step} and \obsref{d:marked}

\subsubsection{Warmup Algorithm}
\alglab{warmup}
    \begin{compactenum}
    \item Let $H_1$ be a copy of $H$, and set $\mathrm{MARKED}\leftarrow \phi$
    
    \medskip
    
    \item \textbf{While} $H_1$ contains a size $(k-1)$ subset $S$ such that 
    $N_{H_1}(S)>t$:
    \medskip
    \begin{compactenum}[\qquad(A)]
    
    		\item \textbf{If} $N_{H_1}(S)$ contains an edge:
        
        \begin{compactenum}[\qquad]
            \item Delete from $H_1$ all edges that completely contain $S$. Also, add 
            $S$ to $\mathrm{MARKED}$.
        \end{compactenum}
        
        \medskip
        
        \item \textbf{If} $N_{H_1}(S)$ is an independent set:
        
        \begin{compactenum}[\qquad]
            \item Use $2$ new colors, color $S$ one color and $N_{H_1}(S)$ the other, 
            remove $S\cup N_{H_1}(S)$ and all edges containing any of these vertices 
            from $H_1$.
        \end{compactenum}
        
    \end{compactenum}
    
    \item[] \textbf{End~While}
    
    \bigskip 
    
    \item Let $F$ be the $(k-1)$-uniform hypergraph whose vertex set is the union 
    of the sets in $\mathrm{MARKED}$, and whose edge set is $\mathrm{MARKED}$. 
    Using $2$ new colors, properly $2$-color the vertices of $F$ using the $2$-coloring 
    algorithm for discrepancy $0$ hypergraphs. Remove these vertices and all 
    involved edges, from $H_1$.
    
    \bigskip    
    
    \item \textbf{Output} the partial coloring of $H$ and the residual graph $H_1$ of 
    max-degree $t$.
\end{compactenum}

\bigskip

\subsubsection{The Algorithm \emph{LDDegreeReduce}}
\alglab{SCDegreeReduce}

\begin{compactenum}
    \item Let $H_1$ and $H_2$ be copies of $H$, $\mathrm{MARKED}\leftarrow \phi$ 
    and $T\leftarrow \phi$.    
    \medskip
    
    \item \textbf{While} $H_2$ contains a size $(k-1)$ subset $S$ of vertices, such 
    that $|N_{H_2}(S)|>t$:
    \medskip
    
    \begin{compactenum}[\qquad(A)]
        \item Delete $N_{H_2}(S)$ and all edges involving these vertices, from $H_2$.
        
        \medskip
        
        \item \textbf{If} $N_{H_1}(S)$ contains an edge:
        
        \begin{compactenum}[\qquad]
            \item Delete from $H_1$ all edges that completely contain $S$. Also, add 
            $S$ to $\mathrm{MARKED}$.
        \end{compactenum}
        
        \medskip
        
        \item \textbf{If} $N_{H_1}(S)$ is an independent set:
        
        \begin{compactenum}[\qquad]
            \item Add $S$ to $T$. 
        \end{compactenum}
        
        \medskip
    
        \item For every size $c$ subset $\overline{V} = \brc{S'_1,\dots ,S'_c}$ of $T$:
        
        \bigskip

        \item[] Fix two new colors $c_1,c_2$.\\ 
        For every possible assignment of $c_1,c_2$ to $\overline{V}$, such that 
        each $S'_i$ has discrepancy $2$, (We define $bias(S'_i) = c_1$ (resp. $c_2$) 
        for coloring bias towards $c_1$ (resp. $c_2$)):
        
        \bigskip
        
            \begin{compactenum}[\qquad(D1)]
                \item For $i=1,2$, let $C_i=\brc{u \sep{S'\in \overline{V},~bias(S')=c_i,~u\in 
                N_{H_1}(S')}}$. Then color $N_{H_1}(C_1)$ with just $c_2$ and $N_{H_1}(C_2)$ with 
                just $c_1$.
                \medskip
            
                \item \textbf{If} no edge is monochromatic: 
            
                \begin{compactenum}[\qquad]
                    \item Stick with this proper $2$-coloring of the vertices in $\overline{V},N_{H_1}(\overline{V})$.
                    
                    \item Remove $\overline{V}$ from $T$, i.e. $T\leftarrow T \setminus \overline{V}$
                    
                    \item Remove $\bigcup_{i} (S'_i\cup N_{H_1}(S'_i))$ and all edges containing any of these 
                    vertices, from $H_1$ and $H_2$, and stop iterating through assignments of $\overline{V}$.
                \end{compactenum}
            
                \medskip            
            
                \item \textbf{If} some edge is monochromatic: 
            
                \begin{compactenum}[\qquad]
                    \item Discard the coloring and continue iterating through assignments of $\overline{V}$.
                \end{compactenum}
            \end{compactenum}
        
    \end{compactenum}
    
    \item[] \textbf{End~While}
    
    \bigskip    
    
    \item For every subset $B$ of $T$ of size less than $c$:
    
    \begin{compactenum}[\qquad(1)]
        \item Let $A\leftarrow T\setminus B$.
        
        \medskip
        
        \item Using two new colors, run the proper $2$-coloring algorithm for discrepany zero hypergraphs 
        on the $(k-1)$-uniform hypergraph whose edge set is $A$.
        
        \medskip
        
        \item Using two new colors, iterate through all assignments of $B$, and attempt to $2$-color 
        $N_{H_1}(B)$ just as in Step (D1).
        
        \medskip
        
        \item \textbf{If} both colorings succeed: 
            
        \begin{compactenum}[\qquad]
            \item Stick with this proper $2$-coloring of the vertices in $T,N_{H_1}(B)$.
            \medskip
                    
            \item Remove from $H_1$ the vertices 
            $
            \bigcup_{S'\in A} S' \text{ and } \bigcup_{S'\in B} (S'\cup N_{H_1}(S'))
            $\\ 
            and all edges involving any of these vertices,
            and stop iterating through subsets of $T$.
        \end{compactenum}

        \medskip            
            
        \item \textbf{If} either coloring fails: 
            
        \begin{compactenum}[\qquad]
            \item Discard the coloring and continue iterating through subsets of $T$.
        \end{compactenum}
    \end{compactenum}
    
    \bigskip
    
    \item \textbf{Output} the proper partial coloring of $H$ and the residual graph $H_1$ of 
    max-degree ${n-1 \choose k-2}t$.
\end{compactenum}

\subsubsection{The Result.}

\begin{theorem}
\thmlab{lddr}
    For odd $k$, let $H=\pth{V,E\subseteq {V \choose k}}$ be a $k$-uniform discrepancy $1$ hypergraph with 
    $n$ vertices. \algref{SCDegreeReduce} partially colors $H$ in $n^{c+O(1)}$ time, with at most 
    $2n/ct$ colors, such that:
    \medskip
    
    \begin{compactenum}
        \item The subgraph of $H$ induced by the colored vertices has no monochromatic vertices.
        
        \medskip
        
        \item The subgraph of $H$ induced by the uncolored vertices has maximum degree 
        ${n-1 \choose k-2}t$.
    \end{compactenum}
\end{theorem}

\begin{proof}
	The proof goes very similarly to that of \thmref{scdr}, thus we just 
	state the key observations required to complete the proof.
	\medskip
	
	\begin{compactenum}[(A)]
		\item In any discrepancy $1$ coloring of $H$, any size $k-1$ 
		set $S'$ either has discrepancy $2$, or discrepancy $0$.
		
		\medskip
	
		\item Consider any discrepancy $1$ coloring of $H$. If a size $k-1$ 
		set $S'$ has discrepancy $2$, then $N(S')$ is monochromatic.
		
		\medskip
	
		\item At the end of any iteration of Step 2., 
		there is no size $c$ subset of $T$ such that every set in the 
		subset has discrepancy $2$ in any discrepancy $1$ coloring of $H$.
		
		\medskip
		
		\item When we reach Step 3., at least $|T|-c$ sets in $T$, all 
		have discrepancy $0$ in EVERY discrepancy $1$ coloring of $H$.
		
	\end{compactenum}	
	~
\end{proof}

\subsection{Degree Reduction under Rainbow Colorability}
\seclab{rcdr}

Now, the equivalent algorithm in the case of rainbow colorability is virtually identical to 
that of \secref{lddr}. Thus we merely state the result.

\begin{theorem}
\thmlab{rcdr}
    Let $H=\pth{V,E\subseteq {V \choose k}}$ be a $k$-uniform $(k-1)$-rainbow colorable hypergraph 
    with $n$ vertices. \algref{SCDegreeReduce} partially colors $H$ in $n^{c+O(1)}$ time, with at 
    most $(k-1)n/ct$ colors, such that:
    \medskip
    
    \begin{compactenum}
        \item The subgraph of $H$ induced by the colored vertices has no monochromatic vertices.
        
        \medskip
        
        \item The subgraph of $H$ induced by the uncolored vertices has maximum degree 
        ${n-1 \choose k-2}t$.
    \end{compactenum}
\end{theorem}

\bibliographystyle{alpha}
\bibliography{hcup_arxiv}

\end{document}